\documentclass[
reprint,
superscriptaddress,
 amsmath,amssymb,
 floatfix,
 aps,
 pra,
]{revtex4-2}

\usepackage{graphicx}
\usepackage[svgnames]{xcolor}
\usepackage{dcolumn}
\usepackage{bm}

\usepackage[
     colorlinks        = true,
     unicode           = true,
     pdfstartview      = FitV,
     linktocpage       = true,
     linkcolor         = OrangeRed,
     citecolor         = MediumSeaGreen,
     urlcolor          = RoyalBlue,
     bookmarks         = true,
     bookmarksnumbered = true,
     breaklinks=true,
     pdftitle={Fast-forwardability of Qubit-mapped Fermion models based on Cartan decomposition},
     pdfauthor={}
]{hyperref}

\usepackage{amsmath}
\usepackage{amsthm}
\usepackage{amssymb}
\usepackage{amsfonts}
\usepackage{cases}
\usepackage{enumitem}

\theoremstyle{plain}
\newtheorem{thm}{Theorem}[section]
\newtheorem{lem}[thm]{Lemma}
\newtheorem{prop}[thm]{Proposition}

\theoremstyle{definition}

\newcommand{\qunasys}{QunaSys Inc., Aqua Hakusan Building 9F, 1-13-7 Hakusan, Bunkyo, Tokyo 113-0001, Japan}
\newcommand{\fujitsu}{Quantum Laboratory, Fujitsu Research, Fujitsu Limited., 4-1-1 Kamikodanaka, Nakahara, Kawasaki, Kanagawa 211-8588, Japan}

\begin{document}

\preprint{APS/123-QED}

\title{Fast-forwardability of Qubit-mapped Fermion models based on Cartan decomposition}

\author{Yuichiro Hidaka}
\email{hidaka@qunasys.com}
\affiliation{\qunasys}

\author{Shoichiro Tsutsui}
\affiliation{\qunasys}

\author{Shota Kanasugi}
\affiliation{\fujitsu}

\author{Norifumi Matsumoto}
\affiliation{\fujitsu}

\author{Kazunori Maruyama}
\affiliation{\fujitsu}

\author{Hirotaka Oshima}
\affiliation{\fujitsu}

\date{\today}

\begin{abstract}
We study the Hamiltonian algebra of qubit-mapped interacting fermion models and their fast-forwardability. We prove that the dimension of the Hamiltonian algebra of the fermion model with single-site Coulomb interaction is bounded from below by the exponential function of the number of sites, and the circuit depth of the Cartan-based fast-forwarding method for such a model also exhibits the same scaling. We apply this proposition to the Anderson impurity model and the Hubbard model and show that the dimension of the Hamiltonian algebra of these models scales exponentially with the number of sites. These behaviors of the Hamiltonian algebras imply that, under fermion-qubit mappings that map Majorana operators to single-term Pauli strings, the qubit models obtained from these fermion models cannot be efficiently simulated using the Cartan-based fast-forwarding method.
\end{abstract}

\maketitle

\section{\label{sec:Introduction}Introduction}
Recently, quantum computing has attracted significant attention from researchers due to its potential to dramatically reduce the computation time for certain problems compared to classical computers. Quantum computing is investigated for applications in various realms of research, ranging from quantum chemistry (see, e.g., Refs.~\cite{alan2005quantumchem,cao2019quantumchemreview,bauer2020quantumchemreview,mcardle2020quantumchemreview,motta2022emergingquantumchem}) to financial engineering (see, e.g., Refs.~\cite{herman2022surveyfinance,bouland2020prospectsfinance}). The application to solid-state physics~\cite{Feynman1999SimulatingPhysicswith} is not the exception. Many researchers in solid-state physics focus on calculating the ground state and its energy of the given crystal Hamiltonian. 

Quantum phase estimation (QPE)~\cite{nielsenchuang2010quantum} is one of the most powerful quantum algorithms to calculate the eigenenergies of a given Hamiltonian. Such calculation strongly depends on the method used to approximate the time evolution operator. Generally, since the terms of the Hamiltonian written in Pauli strings do not commute with each other, there are various approaches to approximating the time evolution operator such as Trotter decomposition~\cite{trotter1959product,suzuki1976generalized,lloyd1996universalquantumsimulator,sornborger1999higherordermethodtrotter,nielsenchuang2010quantum}, variational methods~\cite{Li2017variationalquantumsimulation,Yuan2019theoryofvariational,heya2023subspacevariational,endo2020vqsgeneralprocesses,benedetti2021hardwareefficientVQS,lin2021realimaginarytimeevolution,berthusen2022qdsimulationbyNISQ,mizuta2022LVQC,akahoshi2024compilation,kanasugi2023compilationgreen,kanasugi2024subspace}, qDRIFT~\cite{campbell2019qdrift}, and qubitization method~\cite{Low2019hamiltonian,babbush2018lineartcomplexity}. 

Furthermore, fast-forwarding methods~\cite{atia2017fast, chia2023impossibility,Gu2021fastforwarding,berry2007efficient,Lim_2022,dong2024multi} are also the efficient quantum simulation algorithms that allow the gate complexity to be sublinear in the simulation time $T$ and suppress the accumulation of errors. On the other hand, it is also known that such drastic circuit compression may not be applied to general models due to the no-fast-forwarding theorem~\cite{atia2017fast, chia2023impossibility}. Therefore, it is important to clarify the distinction between the models that are fast-forwardable and those that are not. For example, in studies of the variational fast-forwarding method~\cite{cirstoiu2020variational,gibbs2022long,commeau2020variational}, it has been empirically known that the asymptotic scaling of circuit depth for interacting fermion systems appears to scale exponentially with the number of sites. However, this observation lacks a rigorous mathematical proof. Hence, we focus on the Cartan-based fast-forwarding method~\cite{kokcu2022cartanfastforward} to explore the potential for circuit depth reduction.

In the Cartan-based fast-forwarding method, the time evolution operator is implemented based on the Cartan decomposition of matrices. Whether we can efficiently implement the time evolution operator by the Cartan-based fast-forwarding method depends on the dimension of the Lie algebra generated by the set of the Pauli strings in the Hamiltonian, which we call {\it Hamiltonian algebra}~\cite{d2021introduction,kokcu2022cartanfastforward} or also known as dynamical Lie algebra. 
If the dimension of the Hamiltonian algebra scales polynomially with the number of sites, the circuit depth for the Cartan-based fast-forwarding method exhibits the same scaling.

The Hamiltonian algebra for quantum spin systems in one dimension such as XY model and the Heisenberg model is well studied~\cite{kokcu2022cartanfastforward}.
More generally, the structure of the Hamiltonian algebra of two-local qubit Hamiltonians is classified by the form of generators~\cite{wiersema2024classification,kokcu2024dlagenerallattice}. However, we cannot apply these discussions to Hamiltonians of fermionic systems, where we transform the creation and annihilation operators into the qubit operators.  This is because the resulting qubit Hamiltonians generally break the two-locality of the qubit operators (except for the special cases). 

In this paper, we investigate the dimension of the Hamiltonian algebra of qubit-mapped fermion systems on arbitrary lattice. First, we prove that the dimension of the Hamiltonian algebra of the general free fermion system exhibits a polynomial scaling (specifically, a quadratic scaling)  with the system size. Then, we consider the fermion model with single-site Coulomb interaction. We prove that the Hamiltonian algebra of this model scales exponentially with the system size, unlike the case of the free fermion system. Then, applying these results, we evaluate the dimension of the Hamiltonian algebra of the Anderson impurity model and the Hubbard model. 

This paper is organized as follows: Sec.~\ref{sec:Hamiltonian Algebra} introduces the necessary background of the Lie algebra and the definition of the Hamiltonian algebra. Sec.~\ref{sec:Main Result} and Sec.~\ref{sec:algebraicfastforwarding} present and prove our main theorems for the Hamiltonian algebra of fermion systems, discussing the circuit depth of the Cartan-based fast-forwarding method of the interacting fermion model. Then, we generalize the discussions in Sec.~\ref{sec:Main Result} and Sec.~\ref{sec:algebraicfastforwarding} to other fermion-qubit mappings in Sec.~\ref{sec:generalmappings}. We apply these propositions to the specific interacting fermion models in Sec.~\ref{sec:application}. The applicability of our results is further discussed in Sec.~\ref{sec:exceptions}. with concluding remarks and discussion in Sec.~\ref{sec:conclusionanddiscussion}.

\section{\label{sec:Hamiltonian Algebra}Hamiltonian Algebra}

We first review the fundamentals of Lie algebra and Hamiltonian algebra. A Lie algebra $\mathfrak{g}$ is defined as a vector space in which a Lie bracket $[\cdot,\cdot]:\mathfrak{g}\times\mathfrak{g}\to\mathfrak{g}$ is defined. In this paper, we restrict the Lie bracket to the commutator, $[A, B]=AB-BA$, where $A, B\in\mathfrak{g}$. A subalgebra of the Lie algebra $\mathfrak{g}$ is defined as a subspace of $\mathfrak{g}$ which is closed under the commutator.

Given a set $A=\{g_1,g_2,\cdots,g_n\}\subset\mathfrak{g}$, the subalgebra $\mathfrak{g}(A)$ generated by $A$ is the subalgebra of $\mathfrak{g}$ satisfying the following two conditions: (i) all elements in $A$ belong to $\mathfrak{g}(A)$, (ii) $\mathfrak{g}(A)$ is the minimal subalgebra among those of $\mathfrak{g}$ satisfying the condition (i). The algebra $\mathfrak{g}(A)$ is obtained by the following procedures. We begin with the vector space $S_0=\mathrm{Span}_{\mathbb{R}}(A)$, where $S_0$ is spanned by the elements of $A$ with real coefficients. Then we take the commutator of the elements in $S_0$ and add them to the basis to obtain the vector space $S_1$, which is spanned by the extended basis, $S_0\subset S_1$. By repeating this operation until the vector space becomes no longer larger, we obtain the generated Lie algebra $\mathfrak g(A)$. The Lie algebra generated by the set $A$ is the real vector space spanned by the elements in the form of $[g_{i_k},[g_{i_{k-1}},\cdots [g_{i_1},g_{i_0}]\cdots]]$, where $k$ is a non-negative integer.

We consider the non-trivial $n$-qubit Pauli strings, which are the elements of the set $\{I, X, Y, Z\}^{\otimes n}$ except for the identity matrix $I^{\otimes n}$. These Pauli strings multiplied by the imaginary unit $i$ form the Lie algebra $\mathfrak{su}(2^n)$ by the linear combination, \textit{i.e.},
\begin{align}
    \mathfrak{su}(2^n)=\mathrm{Span}_{\mathbb{R}}\left\{iP_j\middle|P_j:\text{Pauli strings}\right\}.
\end{align}
The Lie algebra $\mathfrak{su}(2^n)$ admits the structure of the inner product: $(A,B)=\mathrm{Tr}[A^{\dagger}B]/2^n$ ($A,B\in\mathfrak{su}(2^n)$). Pauli strings form the normalized and orthogonal basis of $\mathfrak{su}(2^n)$ with respect to the inner product. Another important property of the Pauli strings is that the commutator of two non-commuting Pauli strings yields another Pauli string: $[iP,\ iP']=\pm 2iP''$, where $P$, $P'$, and $P''$ are Pauli strings. For simplicity, we omit the coefficient $i$ in the following discussion and write the commutation relation as $[P,\ P']\propto P''$ since we do not need information on coefficients when discussing the dimension of the subalgebra.

When a certain set of Pauli strings $\mathcal{S}$ is given, the Lie algebra $\mathfrak{g}(\mathcal{S})$ generated by $\mathcal{S}$ is defined as the smallest subalgebra of $\mathfrak{su}(2^n)$ which includes $\mathcal{S}$. In this case, the generated Lie algebra is in the form $\mathfrak{g}(\mathcal{S})=\mathrm{Span}_{\mathbb{R}}(\mathcal{P})$, where $\mathcal{P}(\supset\mathcal{S})$ is a certain set of Pauli strings. Then the dimension of this Lie algebra is $|\mathcal{P}|$. The important properties of the generated Lie algebra are summarized in the following proposition:

\begin{prop}\label{prop:generatedliealg}
(1) If the sets of Pauli strings $\mathcal{S}_1$ and $\mathcal{S}_2$ satisfy $\mathcal{S}_1\subset\mathcal{S}_2$, the Lie algebras generated by these sets satisfy $\mathfrak{g}(\mathcal{S}_1)\subset\mathfrak{g}(\mathcal{S}_2)$, hence $\text{dim}(\mathfrak{g}(\mathcal{S}_1))\le\text{dim}(\mathfrak{g}(\mathcal{S}_2))$.\\
(2) For any set of Pauli strings $\mathcal{S}$, the Lie algebra $\mathfrak{g}(\mathcal{S})$ generated by the set $\mathcal{S}$ satisfies $|\mathcal{S}|\le\text{dim}(\mathfrak{g}(\mathcal{S}))$.
\end{prop}
\begin{proof}
(1) Let $\mathcal{S}_1=\{P_1,P_2,\cdots,P_m\}$. $\mathfrak{g}(\mathcal{S}_1)$ is spanned by $[P_{i_1},[P_{i_2},[\cdots ,[P_{i_{r-1}},P_{i_r}]\cdots]]]$. By the assumption, $P_{i_1}, P_{i_2},\cdots, P_{i_r}$ are also elements of $\mathfrak{g}(\mathcal{S}_2)$. Since $\mathfrak{g}(\mathcal{S}_2)$ is closed under the commutator, terms such as $[P_{i_1},[P_{i_2},[\cdots ,[P_{i_{r-1}},P_{i_r}]\cdots]]]$ are also elements of $\mathfrak{g}(\mathcal{S}_2)$. Hence $\mathfrak{g}(\mathcal{S}_1)\subset\mathfrak{g}(\mathcal{S}_2)$ is satisfied.\\
(2) Since $\mathfrak{g}(\mathcal{S}) = \text{Span}_{\mathbb{R}} (\mathcal{P})$ for some $\mathcal{P}\supseteq S$, the dimension of $\mathfrak{g}(\mathcal{S})$ satisfies $\text{dim}(\mathfrak{g}(\mathcal{S})) = |\mathcal{P}| \ge |\mathcal{S}|$.
\end{proof}

When the Hamiltonian of an $n$-qubit system is given in the form of 
\begin{align}
    H=\sum_{i=1}^L c_iP_i, 
\end{align}
where $P_i$ are the Pauli strings and $c_i$ are real-valued coefficients, the Hamiltonian algebra of $H$ is the Lie algebra generated by the set of Pauli strings in the Hamiltonian $\{P_1, P_2,\cdots, P_L\}$. We simply denote the Hamiltonian algebra of $H$ as $\mathfrak{g}(H)$ in the following.

\section{\label{sec:Main Result}Main Result}
In this section, we discuss the Hamiltonian algebra of the spin-1/2 fermion models. In Secs.~\ref{sec:Main Result} and ~\ref{sec:algebraicfastforwarding}, we focus on the Jordan-Wigner transformation. The general cases are discussed in Sec.~\ref{sec:generalmappings}.

\subsection{\label{subsec:Fermion Model}Fermion Models and Jordan-Wigner Transformation}
We consider the general free fermion model,
\begin{align}\label{eq:freefermion}
H_0=\sum_{\sigma=\uparrow,\downarrow}\sum_{(k,l)\in G}t_{kl}(c^{\dagger}_{k\sigma}c_{l\sigma}+\text{h.c.}),
\end{align}
where $\sigma$ is a spin and $c$ and $c^{\dagger}$ are the fermion annihilation and creation operators respectively. In Eq.~\eqref{eq:freefermion}, $G$ denotes a connected graph with $N$ vertices and $(k,l)\in G$ means that the vertices $k$ and $l$ are connected by an edge in $G$. Here, a connected graph is defined as a graph where, for any two vertices, a path connecting them exists. For example, a one-dimensional chain and a two-dimensional square lattice are both connected graphs. Examples of connected graphs are shown in Fig.~\ref{fig:connectedgraphs}. When the graph $G$ is disconnected, the structure of the Hamiltonian algebra reduces to the direct sum of the connected components. Hence, we restrict our attention to the case where the graph representing the hopping terms is connected. 

\begin{figure}[t]
\includegraphics[scale=0.5]{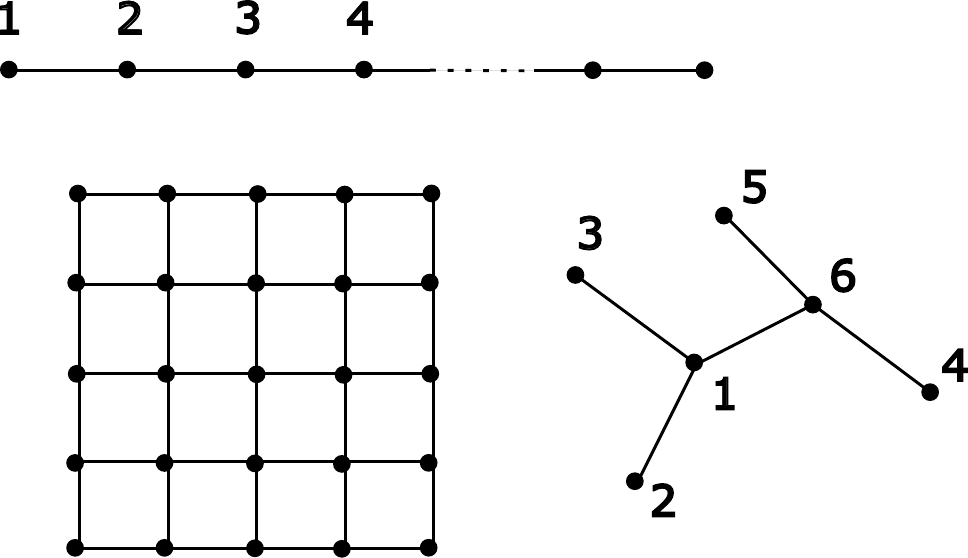}
\caption{\label{fig:connectedgraphs} Examples of connected graphs.}
\end{figure}

The Jordan-Wigner transformation of the fermion operators is given by the following equations:
\begin{align}
    \label{eq:jw_fermion_def_dag}c^{\dagger}_{k\sigma}=Z_1Z_2\cdots Z_{i-1}\left(\frac{X_i + iY_i}{2}\right),\\\label{eq:jw_fermion_def}
    c_{k\sigma}=Z_1Z_2\cdots Z_{i-1}\left(\frac{X_i - iY_i}{2}\right),\\
    \nonumber
\end{align}
where the fermion indices $(k,\sigma)$ are aligned in a row and assigned a new index $i$ for the qubit configuration. We assume that the spin-up fermion sites correspond to the sites belonging to $S_{\text{up}}=\{u(0),u(1),\cdots,u(N-1)\}$ in the qubit configuration, while the spin-down fermion sites correspond to $S_{\text{down}}=\{d(0),d(1),\cdots,d(N-1)\}$. These sets satisfy the following relations:
\begin{align}
    S_{\text{up}}\cup S_{\text{down}}&=\{0,1,2,\cdots,2N-2,2N-1\},\\
    S_{\text{up}}\cap S_{\text{down}}&=\emptyset.    
\end{align}
Using the Jordan-Wigner transformation, the free fermion model~\eqref{eq:freefermion} is transformed to
\begin{align}\label{eq:freejw}
    H_{0,\text{JW}}=&\frac12\sum_{(k,l)\in G}t_{u(k)u(l)}(\overrightarrow{X_{u(k)}X_{u(l)}}+\overrightarrow{Y_{u(k)}Y_{u(l)}})\nonumber\\
    &+\frac12\sum_{(k,l)\in G}t_{d(k)d(l)}(\overrightarrow{X_{d(k)}X_{d(l)}}+\overrightarrow{Y_{d(k)}Y_{d(l)}}),
\end{align}
where we define $\overrightarrow{Q_iQ_j}$ ($Q_i=X_i$, $Y_i$, or $Z_i$) as
\begin{align}
    \overrightarrow{Q_iQ_j}=
    \begin{cases}
        Q_iZ_{i+1}Z_{i+2}\cdots Z_{j-1}Q_j\  (i<j);\\
        Q_jZ_{j+1}Z_{j+2}\cdots Z_{i-1}Q_i\  (j<i).
    \end{cases}
\end{align}
On the other hand, we also consider the fermion model with only single-site Coulomb interaction, which is given by
\begin{align}\label{eq:oneintfermion}
    H_{1}=&H_0+4U\left(n_{0\uparrow}-\frac12\right)\left(n_{0\downarrow}-\frac12\right),
\end{align}
where $U\neq0$ and $n_{i\sigma}=c^{\dagger}_{i\sigma}c_{i\sigma}$ is the number operator at the site $i$. We can consider the corresponding models where the interaction term is not at the site 0. However, since we consider the Jordan-Wigner transformation with general site ordering, and due to the symmetry, the discussion reduces to the case where the interaction term is at the site 0. Since Eqs.~\eqref{eq:jw_fermion_def_dag} and~\eqref{eq:jw_fermion_def} yield  $2c^{\dagger}_{i'\sigma}c_{i'\sigma}-1=Z_{i}$, the Jordan-Wigner transformed Hamiltonian of the interacting fermion model \eqref{eq:oneintfermion} is obtained as
\begin{align}\label{eq:oneintjw}
    H_{1,\text{JW}}=&H_{0,\text{JW}}+UZ_{u(0)}Z_{d(0)}.
\end{align}
In Eqs.~\eqref{eq:oneintfermion} and \eqref{eq:oneintjw}, the site indices $(0,\uparrow)$ and $(0,\downarrow)$ correspond to $u(0)$ and $d(0)$, respectively. We impose the relations $u(1)<u(2)<\cdots<u(N-1)$ and $d(1)<d(2)<\cdots<d(N-1)$, which do not involve $u(0)$ and $d(0)$. 

The main result of our paper is that the Hamiltonian algebra of the free fermion model in Eq.~\eqref{eq:freejw} and that of the fermion model with single-site Coulomb interaction  in Eq.~\eqref{eq:oneintjw} satisfy the following properties:

\begin{thm}\label{thm:hamalgdim}
\begin{enumerate}[label=(\arabic*)]
    \item The dimension of the Hamiltonian algebra of  Eq.~\eqref{eq:freejw} satisfies
    \begin{align}
        \text{dim}(\mathfrak{g}(H_{0,\text{JW}}))\le 2N(2N-1).
    \end{align}
    Therefore, it scales polynomially with $N$.
    \item The dimension of the Hamiltonian algebra of  Eq.~\eqref{eq:oneintjw} satisfies
    \begin{align}\label{eq:dim_ineq_hamilang_interact}
        \text{dim}(\mathfrak{g}(H_{1,\text{JW}}))\ge 2^{N-1}.
    \end{align}
    Therefore, it scales exponentially with $N$.
\end{enumerate}
\end{thm}\noindent
The first statement has already been discussed essentially in previous studies~\cite{atia2017fast,Gu2021fastforwarding} in the context of the fast-forwarding method. However, there are differences between the previous and the present statement. For example, we treat the Jordan-Wigner transformed fermion models instead of treating the corresponding models in the fermion representation. Although our proof seems to be based on an idea similar to that in Ref.~\cite{Gu2021fastforwarding}, we restate the first statement of the theorem because, in addition to the above reasons, we intend to compare this result with the second statement and elucidate the critical term that pushes the dimension of the Hamiltonian algebra into the exponential order. We prove Theorem~\ref{thm:hamalgdim} in the following subsections.

\subsection{\label{subsec:proof free}Proof of Theorem: Free Fermion}

We prove the statement in Theorem~\ref{thm:hamalgdim}~(1). The Hamiltonian algebra of the Hamiltonian in Eq.~\eqref{eq:freejw} is the direct sum of the spin-up and spin-down parts of the Hamiltonian, \begin{align}
\mathfrak{g}(H_{0,\text{JW}})=\mathfrak{g}(H_{0,JW,u})\oplus\mathfrak{g}(H_{0,JW,d}),
\end{align}
where 
\begin{align}
H_{0,JW,\sigma}=\frac12\sum_{(k,l)\in G}t_{\sigma(k)\sigma(l)}(\overrightarrow{X_{\sigma(k)}X_{\sigma(l)}}+\overrightarrow{Y_{\sigma(k)}Y_{\sigma(l)}})
\end{align}
for $\sigma=u$, $d$. Here, $\mathfrak{g}(H_{0,JW,u})$ is included in the following vector space:
\begin{align}
    \mathfrak{g}_{\text{large}}&=\text{Span}_{\mathbb{R}}(S_{XX}\cup S_{YY}\cup S_{XY}\cup S_{YX}\cup S_{Z}),
\end{align}
where the subsets such as $S_{XX}$ are given by
\begin{align}
    S_{XX}&=\{\overrightarrow{X_{u(k)}X_{u(l)}}|k,l=0,1,\cdots,N-1\},\\
    S_{YY}&=\{\overrightarrow{Y_{u(k)}Y_{u(l)}}|k,l=0,1,\cdots,N-1\},\\
    S_{XY}&=\{\overrightarrow{X_{u(k)}Y_{u(l)}}|k,l=0,1,\cdots,N-1\},\\
    S_{YX}&=\{\overrightarrow{Y_{u(k)}X_{u(l)}}|k,l=0,1,\cdots,N-1\},\\
    S_{Z}&=\{Z_{u(k)}|k=0,1,\cdots,N-1\}.   
\end{align}
The vector space $\mathfrak{g}_{\text{large}}$ is closed under the commutator. Therefore, by using Proposition~\ref{prop:generatedliealg}, we obtain $\mathfrak{g}(H_{0,JW,u})\subset\mathfrak{g}_{\text{large}}$, and hence 
\begin{align}
    \text{dim}(\mathfrak{g}(H_{0,JW,u}))\le\text{dim}(\mathfrak{g}_{\text{large}})=N(2N-1).
\end{align}
We also obtain the same inequality for $\mathfrak{g}(H_{0,JW,d})$. Therefore, we obtain
\begin{align}
    \text{dim}(\mathfrak{g}(H_{0,\text{JW}}))\le 2N(2N-1).
\end{align}
This proves Theorem~\ref{thm:hamalgdim}~(1).

The Lie algebra $\mathfrak{g}_{\text{large}}$ is isomorphic to the Hamiltonian algebra of the XY model with magnetic field~\cite{kokcu2022cartanfastforward}:
\begin{align}\label{eq:xymodelwithz}
    H=\sum_{k=1}^{N-1}(X_kX_{k+1}+Y_kY_{k+1})+\sum_{k=1}^{N}Z_k.
\end{align}
According to Ref.~\cite{wiersema2024classification}, the Hamiltonian algebra of Eq.~\eqref{eq:xymodelwithz} is isomorphic to $\mathfrak{so}(2N)$. The dimension of the Lie algebra $\mathfrak{so}(2N)$ is $N(2N-1)$, which is consistent with our result. Therefore, the upper bound of the dimension of the Hamiltonian algebra for free fermion models can be intuitively understood as the dimension of $\mathfrak{so}(2N)\oplus\mathfrak{so}(2N)$.

\subsection{\label{subsec:proof interact} Proof of Theorem: Interacting Fermion}
Next, we prove the statement in Theorem~\ref{thm:hamalgdim}~(2). To prove this, we systematically construct distinct $2^n$ Pauli strings in the Hamiltonian algebra $\mathfrak{g}(H_{1,\text{JW}})$ by using non-crossing paths on a ladder diagram that we will introduce below. 

In the main text, we restrict our discussion to the case where $u(0)$ ($d(0)$) is the smallest number in the set $\{u(0),u(1),\cdots,u(N-1)\}$ ($\{d(0),d(1),\cdots,d(N-1)\}$). For general $u(0)$ and $d(0)$, we also obtain the same claims, but we need to modify the proof slightly. The modification is mentioned in Appendix~\ref{app:correction}.

\begin{lem}\label{lem:ladderparts}
    For any $i \in \{1,2,\cdots,N-1\}$, there exists a Pauli operator $P^+_{u(i)}=X_{u(i)}$ or $Y_{u(i)}$ (and similarly, $P^+_{d(i)}=X_{d(i)}$ or $Y_{d(i)}$) such that
    \begin{align}
    \overrightarrow{X_{u(0)}P^+_{u(i)}},\quad \overrightarrow{Y_{u(0)}P^-_{u(i)}}\in \mathfrak{g}(H_{1,\text{JW}}),\\
    \overrightarrow{X_{d(0)}P^+_{d(i)}},\quad \overrightarrow{Y_{d(0)}P^-_{d(i)}}\in \mathfrak{g}(H_{1,\text{JW}}),
    \end{align}
    where $P^-_{k}$ is defined as
    \begin{align}
        P^-_{k}=
        \begin{cases}
            Y_{k}\quad (\text{if $P^+_{k}=X_{k}$});\\
            X_{k}\quad (\text{if $P^+_{k}=Y_{k}$}).
        \end{cases}
    \end{align}
    Furthermore, for such $P^{\pm}_k$'s, the following relations also hold:
    \begin{align}
    \overrightarrow{P^+_{u(i)}P^-_{u(i+1)}},\quad \overrightarrow{P^-_{u(i)}P^+_{u(i+1)}},\in\mathfrak{g}(H_{1,\text{JW}}),\\
    \overrightarrow{P^+_{d(i)}P^-_{d(i+1)}},\quad \overrightarrow{P^-_{d(i)}P^+_{d(i+1)}},\in\mathfrak{g}(H_{1,\text{JW}}).\label{eq:PdPdingHJW}
    \end{align}
\end{lem}
\begin{proof}
    Since we assume that the graph $G$ is connected, for any $u(i)$, there exists a path from $u(0)$ to $u(i)$ in the graph $G$. Let the path be
    \begin{align}
        u(0)\to u(j_1)\to u(j_2)\to\cdots \to u(j_m)=u(i).
    \end{align} The Pauli strings $\overrightarrow{X_{u(0)}X_{u(j_1)}}$, $ \overrightarrow{Y_{u(0)}Y_{u(j_1)}}$, and all the terms of the form $\overrightarrow{X_{u(j_k)}X_{u(j_{k+1})}}$ and $ \overrightarrow{Y_{u(j_k)}Y_{u(j_{k+1})}}$ are the elements of  $\mathfrak{g}(H_{1,\text{JW}})$ because they appear in the terms of the Hamiltonian \eqref{eq:oneintjw}.

    We replace $i$ in the claim with $j_k$ and prove it inductively on $k=1,\cdots,m$. First, in the case of $k=1$, since $\overrightarrow{X_{u(0)}X_{u(j_1)}}$ and $\quad \overrightarrow{Y_{u(0)}Y_{u(j_1)}}$ belong to $ \mathfrak{g}(H_{1,\text{JW}})$, the claim holds by setting $P_{u(j_1)}=X_{u(j_1)}$. 
    Then, we consider the case of $k=2$. In the case $u(j_1)<u(j_2)$, $\overrightarrow{X_{u(0)}X_{u(j_1)}}$ and $\overrightarrow{Y_{u(j_1)}Y_{u(j_2)}}$ anticommute, while $\overrightarrow{Y_{u(0)}Y_{u(j_1)}}$ and $\overrightarrow{X_{u(j_1)}X_{u(j_2)}}$ also anticommute. Hence we have
    \begin{align}
    [\overrightarrow{X_{u(0)}X_{u(j_1)}},\overrightarrow{Y_{u(j_1)}Y_{u(j_2)}}]\propto \overrightarrow{X_{u(0)}Y_{u(j_2)}}\in\mathfrak{g}(H_{1,\text{JW}}),\\
    [\overrightarrow{Y_{u(0)}Y_{u(j_1)}},\overrightarrow{X_{u(j_1)}X_{u(j_2)}}]\propto \overrightarrow{Y_{u(0)}X_{u(j_2)}}\in\mathfrak{g}(H_{1,\text{JW}}).
    \end{align}
    On the other hand, in the case $u(j_1)>u(j_2)$, $\overrightarrow{X_{u(0)}X_{u(j_1)}}$ ($\overrightarrow{Y_{u(0)}Y_{u(j_1)}}$) and $\overrightarrow{X_{u(j_1)}X_{u(j_2)}}$ ($\overrightarrow{Y_{u(j_1)}Y_{u(j_2)}}$) anticommute. Hence, by taking the commutator of these Pauli strings, we can show that the Pauli strings $\overrightarrow{X_{u(0)}Y_{u(j_2)}}$ and $\overrightarrow{Y_{u(0)}X_{u(j_2)}}$ belong to $\mathfrak{g}(H_{1,\text{JW}})$. In any case, $\overrightarrow{X_{u(0)}Y_{u(j_2)}},\quad \overrightarrow{Y_{u(0)}X_{u(j_2)}}\in\mathfrak{g}(H_{1,\text{JW}})$ holds and the claim holds by setting $P_{u(j_2)}=Y_{u(2)}$.
    
    Then, by the same discussion, we obtain 
    \begin{align}
        \overrightarrow{X_{u(0)}X_{u(j_3)}},\quad \overrightarrow{Y_{u(0)}Y_{u(j_3)}}\in\mathfrak{g}(H_{1,\text{JW}}),
    \end{align}
    by taking the commutators of the Pauli strings $\overrightarrow{X_{u(0)}Y_{u(j_2)}}$, $\overrightarrow{Y_{u(0)}X_{u(j_2)}}$, $\overrightarrow{X_{u(j_2)}X_{u(j_3)}}$, and $\overrightarrow{Y_{u(j_2)}Y_{u(j_3)}}$.
    
    By using the mathematical induction, for any $k$, we obtain 
    \begin{align}
        \overrightarrow{X_{u(0)}X_{u(j_k)}},\quad \overrightarrow{Y_{u(0)}Y_{u(j_k)}}\in \mathfrak{g}(H_{1,\text{JW}})
    \end{align}
    if the path length $k$ is odd, or
    \begin{align}
        \overrightarrow{X_{u(0)}Y_{u(j_k)}},\quad \overrightarrow{Y_{u(0)}X_{u(j_k)}}\in \mathfrak{g}(H_{1,\text{JW}})
    \end{align}
    if $k$ is even. By setting $k=m$, we obtain the desired result since $j_m=i$. 
    
    The last sentence in the statement follows from 
    \begin{align}
        [\overrightarrow{X_{u(0)}P^+_{u(i)}},\overrightarrow{X_{u(0)}P^+_{u(i+1)}}]\propto \overrightarrow{P^-_{u(i)}P^+_{u(i+1)}},\\
        [\overrightarrow{Y_{u(0)}P^-_{u(i)}},\overrightarrow{Y_{u(0)}P^-_{u(i+1)}}]\propto \overrightarrow{P^+_{u(i)}P^-_{u(i+1)}}.
    \end{align}
    A similar argument also yields Eq.~\eqref{eq:PdPdingHJW}.
\end{proof}
In general, when the path distance from $u(0)$ to $u(i)$, which corresponds to $m$ in this proof, is odd (even), we adopt $P^+_{u(i)}=X_{u(i)}(Y_{u(i)})$. However, when the graph is not bipartite, we can construct both the even- and odd-distance paths from $u(0)$ to $u(i)$. In these cases, $P^+_{u(i)}$ can be chosen as either $X_{u(i)}$ or $Y_{u(i)}$ for any $i$. This choice does not affect the later discussion.

Next, we prove the following lemma which plays a role in constructing the rung parts of the ladder diagram. We still restrict our attention to the condition where $u(0)$ and $d(0)$ are the smallest in the following proof. A more general argument without this restriction is mentioned in Appendix.~\ref{app:correction}.

\begin{lem}\label{lem:rungparts}
    The following Pauli strings belong to the Hamiltonian algebra $\mathfrak{g}(H_{1,\text{JW}})$:
    \begin{align}\label{eq:rungstrings_lem}
    &Z_{d(0)}\cdot Z_{d(i)}\cdot \overrightarrow{P^+_{u(i)}P^-_{u(i+1)}},\\
    &Z_{d(0)}\cdot Z_{d(i)}\cdot\overrightarrow{P^-_{u(i)}P^+_{u(i+1)}},\\
    &Z_{u(0)}\cdot Z_{u(i)}\cdot \overrightarrow{P^+_{d(i)}P^-_{d(i+1)}},\\
    &Z_{u(0)}\cdot Z_{u(i)}\cdot\overrightarrow{P^-_{d(i)}P^+_{d(i+1)}}.
    \end{align}
\end{lem}
\begin{proof}
    From Lemma~\ref{lem:ladderparts}, $\overrightarrow{X_{u(0)}P^+_{u(i)}}$ is an element of $\mathfrak{g}(H_{1,\text{JW}})$. In addition,  $Z_{u(0)}Z_{d(0)}$ is one of the Hamiltonian terms and hence this string belongs to $\mathfrak{g}(H_{1,\text{JW}})$. Then the following Pauli string is also an element of $\mathfrak{g}(H_{1,\text{JW}})$:
    \begin{align}
        [\overrightarrow{X_{u(0)}P^+_{u(i)}},Z_{u(0)}Z_{d(0)}]\propto \overrightarrow{Y_{u(0)}P^+_{u(i)}}\cdot Z_{d(0)}.
    \end{align}
    Moreover, by taking the commutator of this string and $\overrightarrow{X_{d(0)}P^+_{d(i)}}\in\mathfrak{g}(H_{1,\text{JW}})$, we can say that the string 
    \begin{align}
        \overrightarrow{Y_{u(0)}P^+_{u(i)}}\cdot \overrightarrow{Y_{d(0)}P^+_{d(i)}}
    \end{align}
    belongs to $\mathfrak{g}(H_{1,\text{JW}})$. Similarly, we obtain the following strings in $\mathfrak{g}(H_{1,\text{JW}})$:
    \begin{align}
        \overrightarrow{X_{u(0)}P^-_{u(i)}}\cdot \overrightarrow{Y_{d(0)}P^+_{d(i)}},\\
        \overrightarrow{Y_{u(0)}P^+_{u(i)}}\cdot \overrightarrow{X_{d(0)}P^-_{d(i)}},\\
        \overrightarrow{X_{u(0)}P^-_{u(i)}}\cdot \overrightarrow{X_{d(0)}P^-_{d(i)}}.
    \end{align}
    We introduce the following notations:
    \begin{align}\label{eq:Sigmas1}
        \Sigma_1(i)&:=\overrightarrow{Y_{u(0)}P^+_{u(i)}}\cdot \overrightarrow{Y_{d(0)}P^+_{d(i)}},\\
        \Sigma_2(i)&:=\overrightarrow{X_{u(0)}P^-_{u(i)}}\cdot \overrightarrow{Y_{d(0)}P^+_{d(i)}},\\
        \Sigma_3(i)&:=\overrightarrow{Y_{u(0)}P^+_{u(i)}}\cdot \overrightarrow{X_{d(0)}P^-_{d(i)}},\\
        \label{eq:Sigmas4}\Sigma_4(i)&:=\overrightarrow{X_{u(0)}P^-_{u(i)}}\cdot \overrightarrow{X_{d(0)}P^-_{d(i)}}.
    \end{align}
From the above discussion, we can say that the terms in Eqs.~\eqref{eq:Sigmas1} to~\eqref{eq:Sigmas4} belong to $\mathfrak{g}(H_{1,\text{JW}})$. By taking the commutator of $\Sigma_{k}(i)$ and $\Sigma_{l}(i+1)$ ($k,l=1,2,3,4$), we obtain the following relations:
\begin{align}
    \label{eq:ZPPPP1}
    Z_{u(0)}\cdot\overrightarrow{P^-_{u(i)}P^-_{u(i+1)}}\cdot \overrightarrow{P^-_{d(i)}P^+_{d(i+1)}}\in\mathfrak{g}(H_{1,\text{JW}}),\\
    \label{eq:ZPPPP2}Z_{u(0)}\cdot\overrightarrow{P^+_{u(i)}P^+_{u(i+1)}}\cdot \overrightarrow{P^+_{d(i)}P^-_{d(i+1)}}\in\mathfrak{g}(H_{1,\text{JW}}),\\
    \label{eq:ZPPPP3}Z_{d(0)}\cdot\overrightarrow{P^-_{d(i)}P^-_{d(i+1)}}\cdot \overrightarrow{P^-_{u(i)}P^+_{u(i+1)}}\in\mathfrak{g}(H_{1,\text{JW}}),\\
    \label{eq:ZPPPP4}    Z_{d(0)}\cdot\overrightarrow{P^+_{d(i)}P^+_{d(i+1)}}\cdot \overrightarrow{P^+_{u(i)}P^-_{u(i+1)}}\in\mathfrak{g}(H_{1,\text{JW}}).
\end{align}
The above calculation is schematically described in Fig.~\ref{fig:commutation_sigma}. In Fig.~\ref{fig:commutation_sigma}, we represent the positions of the Pauli operators that constitute the operator $\Sigma_{l}(i)$ and $\Sigma_{l}(i+1)$ with blue and red lines, respectively. The circles represent the position contributing to the (anti)commutation relations. Finally, we take the commutators of these strings and $\overrightarrow{P^+_{u(i)}P^-_{u(i+1)}}$ or $\overrightarrow{P^-_{u(i)}P^+_{u(i+1)}}$ or their counterparts with the down-spin index. Then we obtain the strings in Eq.~\eqref{eq:rungstrings_lem} belonging to $\mathfrak{g}(H_{1,\text{JW}})$.
\end{proof}

\begin{figure}[t]
\includegraphics[scale=0.5]{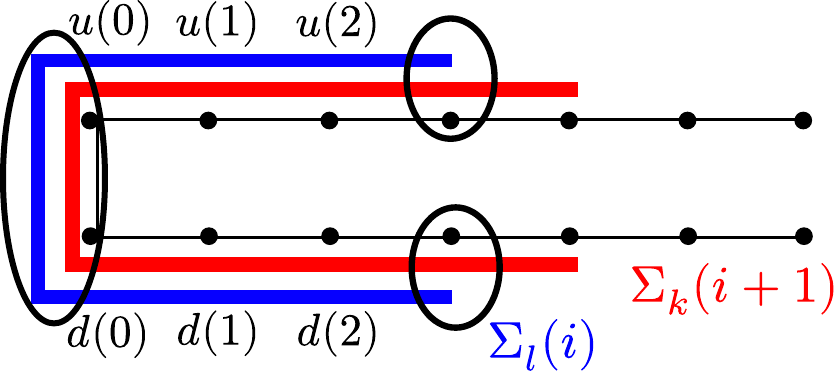}
\caption{\label{fig:commutation_sigma} Schematic description of calculation to obtain Eqs.~\eqref{eq:ZPPPP1} to ~\eqref{eq:ZPPPP4}. The “U-shaped” diagrams represent the positions of the Pauli strings that constitute the operator $\Sigma_l(i)$ and $\Sigma_l(i+1)$. The circles represent the parts contributing to the commutation relations.}
\end{figure}

On the basis of the above lemmas, we now prove Theorem~\ref{thm:hamalgdim}~(2) below. We consider the non-crossing paths from left to right side on a ladder graph, as in Fig.~\ref{fig:ladderdiagram1}. In Fig.~\ref{fig:ladderdiagram1}, we omit the points $u(0)$ and $d(0)$, because we keep the anticommutativity of $Z_{u(0)}\cdot Z_{u(k+1)}\cdot \overrightarrow{P^+_{d(k+1)}P^-_{d(k+2)}}$ or $Z_{d(0)}\cdot Z_{d(k+1)}\cdot \overrightarrow{P^+_{u(k+1)}P^-_{u(k+2)}}$ with the following Pauli operator $Q_k$. For each $k$, the path goes through one of the edges $(u(k),u(k+1))$ or $(d(k),d(k+1))$. Hence, we describe a path by the form $[(i_1,j_1)(i_2,j_2)\cdots(i_{N-1},j_{N-1})]$, where $(i_k,j_k)$ is either $(u(k),u(k+1))$ or $(d(k),d(k+1))$. We show an example of a path in Fig.~\ref{fig:contour_ladder1}. For each non-crossing path on the ladder $[(i_1,j_1)(i_2,j_2)\cdots(i_{N-1},j_{N-1})]$, we assign a Pauli string of the Hamiltonian algebra $\mathfrak{g}(H_{1,\text{JW}})$ in the following way:
\begin{enumerate}
    \item If $(i_1,j_1)=(u(1),u(2))$, we set a Pauli string either $Q_1=\overrightarrow{P^+_{u(1)}P^-_{u(2)}}$ or $Q_1=\overrightarrow{P^-_{u(1)}P^+_{u(2)}}$. If $(i_1,j_1)=(d(1),d(2))$, we set a Pauli string either $Q_1=\overrightarrow{P^+_{d(1)}P^-_{d(2)}}$ or $Q_1=\overrightarrow{P^-_{d(1)}P^+_{d(2)}}$. All these strings belong to $\mathfrak{g}(H_{1,\text{JW}})$ according to Lemma.~\ref{lem:ladderparts}.
    \item We iterate the following operation for $k=1,2,\cdots,N-2$:\\
    Depending on the combination of edges $(i_k,j_k)$ and $(i_{k+1},j_{k+1})$, we generate the new Pauli string $Q_{k+1}$ from $Q_{k}$.
    \begin{enumerate}
        \item\label{enum:operation_uu} $(i_k,j_k)=(u(k),u(k+1))$ and $(i_{k+1},j_{k+1})=(u(k+1),u(k+2))$:\\
        The string $Q_{k}$ anticommutes with one of the strings $\overrightarrow{P^+_{u(k+1)}P^-_{u(k+2)}}$ or $\overrightarrow{P^-_{u(k+1)}P^+_{u(k+2)}}$. By taking the commutator of such anticommuting strings, we obtain, for example,
        \begin{align}
            Q_{k}\cdot\overrightarrow{P^+_{u(k+1)}P^-_{u(k+2)}}\in\mathfrak{g}(H_{1,\text{JW}}).
        \end{align}
        We define the resulting string as $Q_{k+1}$.
        
        \item $(i_k,j_k)=(d(k),d(k+1))$ and $(i_{k+1},j_{k+1})=(d(k+1),d(k+2))$:\\
        As in the case (a), we take the commutator of $Q_{k}$ and either $\overrightarrow{P^+_{d(k+1)}P^-_{d(k+2)}}$ or $\overrightarrow{P^-_{d(k+1)}P^+_{d(k+2)}}$. We define the nonvanishing one in the resulting strings as $Q_{k+1}$.
        
        \item $(i_k,j_k)=(u(k),u(k+1))$ and $(i_{k+1},j_{k+1})=(d(k+1),d(k+2))$:\\
        The string $Q_{k}$ anticommutes with the string $Z_{u(0)}\cdot Z_{u(k+1)}\cdot \overrightarrow{P^+_{d(k+1)}P^-_{d(k+2)}}$. This string belongs to $\mathfrak{g}(H_{1,\text{JW}})$ according to Lemma~\ref{lem:rungparts}. Then we take the commutator of these strings to obtain
        \begin{align}
            Q_{k}\cdot Z_{u(0)}\cdot Z_{u(k+1)}\cdot \overrightarrow{P^+_{d(k+1)}P^-_{d(k+2)}}\in\mathfrak{g}(H_{1,\text{JW}}).
        \end{align}
        We define the resulting string as $Q_{k+1}$.
        
        \item $(i_k,j_k)=(d(k),d(k+1))$ and $(i_{k+1},j_{k+1})=(u(k+1),u(k+2))$:\\
        As in the case (c), we take the commutator of $Q_{k}$ and $Z_{d(0)}\cdot Z_{d(k+1)}\cdot \overrightarrow{P^+_{u(k+1)}P^-_{u(k+2)}}$. We define the resulting string as $Q_{k+1}$.
    \end{enumerate}
\end{enumerate}
We graphically represent the process of step 2. in Fig.~\ref{fig:operation_all}. The reason why these operations succeed is as follows. At step 2. with $k$, if $(i_k,j_k)=(u(k),u(k+1))$ the tentative string $Q_{k}$ consists of 
\begin{itemize}
    \item $P^{\pm}_{u(k+1)}$,
    \item X or Y Pauli matrices at some spin-up sites with indices smaller than $u(k+1)$ (except for $u(0)$),
    \item X or Y Pauli matrices at some spin-down sites with indices smaller than $d(k+1)$ (except for $d(0)$),
    \item Z Pauli matrices at some other sites. 
\end{itemize}
Therefore, $Q_k$ anticommutes with one of the strings $\overrightarrow{P^+_{u(k+1)}P^-_{u(k+2)}}$ or $\overrightarrow{P^-_{u(k+1)}P^+_{u(k+2)}}$. At the same time, $Q_k$ anticommutes with $Z_{u(0)}\cdot Z_{u(k+1)}\cdot \overrightarrow{P^+_{d(k+1)}P^-_{d(k+2)}}$. A similar argument also holds for the case $(i_k,j_k)=(d(k),d(k+1))$.

The resulting string $Q_{N-1}$ obtained by the above procedure involves the Pauli matrices X or Y only at the points $(i_k,j_k)$ where the path transitions occur between the upper and lower lanes. Moreover, $X_{u(1)}$ or $Y_{u(1)}$ ($X_{d(1)}$ or $Y_{d(1)}$) emerges only when the path starts from the upper (lower) lane. Therefore, for each path, we can assign a distinct Pauli string in the Hamiltonian algebra $\mathfrak{g}(H_{1,\text{JW}})$. We have a degree of freedom of choosing one of the two strings at step 1, and for each choice, there are $2^{N-2}$ possible paths on the ladder. Therefore, by enumerating this kind of strings, we obtain the inequality
\begin{align}
    &\text{dim}(\mathfrak{g}(H_{1,\text{JW}}))\nonumber\\
    &\ge 2\cdot(\text{\# of the paths on the ladder})\nonumber\\
    &=2^{N-1}.
\end{align}
This proves Theorem~\ref{thm:hamalgdim}~(2).

\begin{figure}[t]
\includegraphics[scale=0.5]{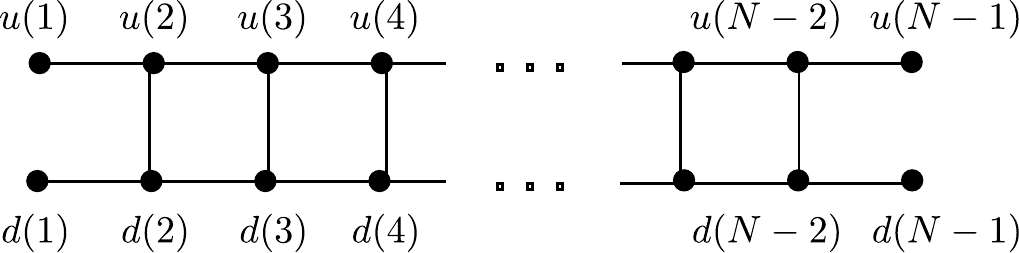}
\caption{\label{fig:ladderdiagram1} 
Ladder graph. The length of the graph is equal to $N-1$.}
\end{figure}

\begin{figure}[t]
\includegraphics[scale=0.6]{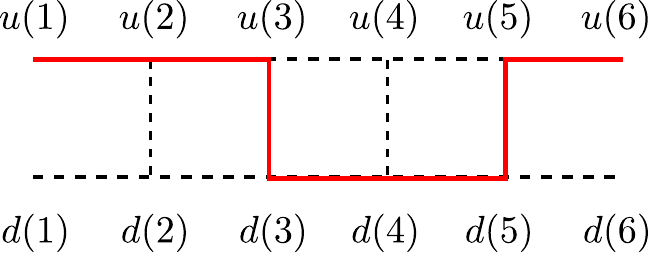}
\caption{\label{fig:contour_ladder1} 
Example of a path of the ladder diagram. In the case of this figure, we represent the path as $[(u(1),u(2))(u(2),u(3))(d(3),d(4))(d(4),d(5))(u(5),u(6))]$}
\end{figure}

\begin{figure}[t]
\includegraphics[scale=0.5]{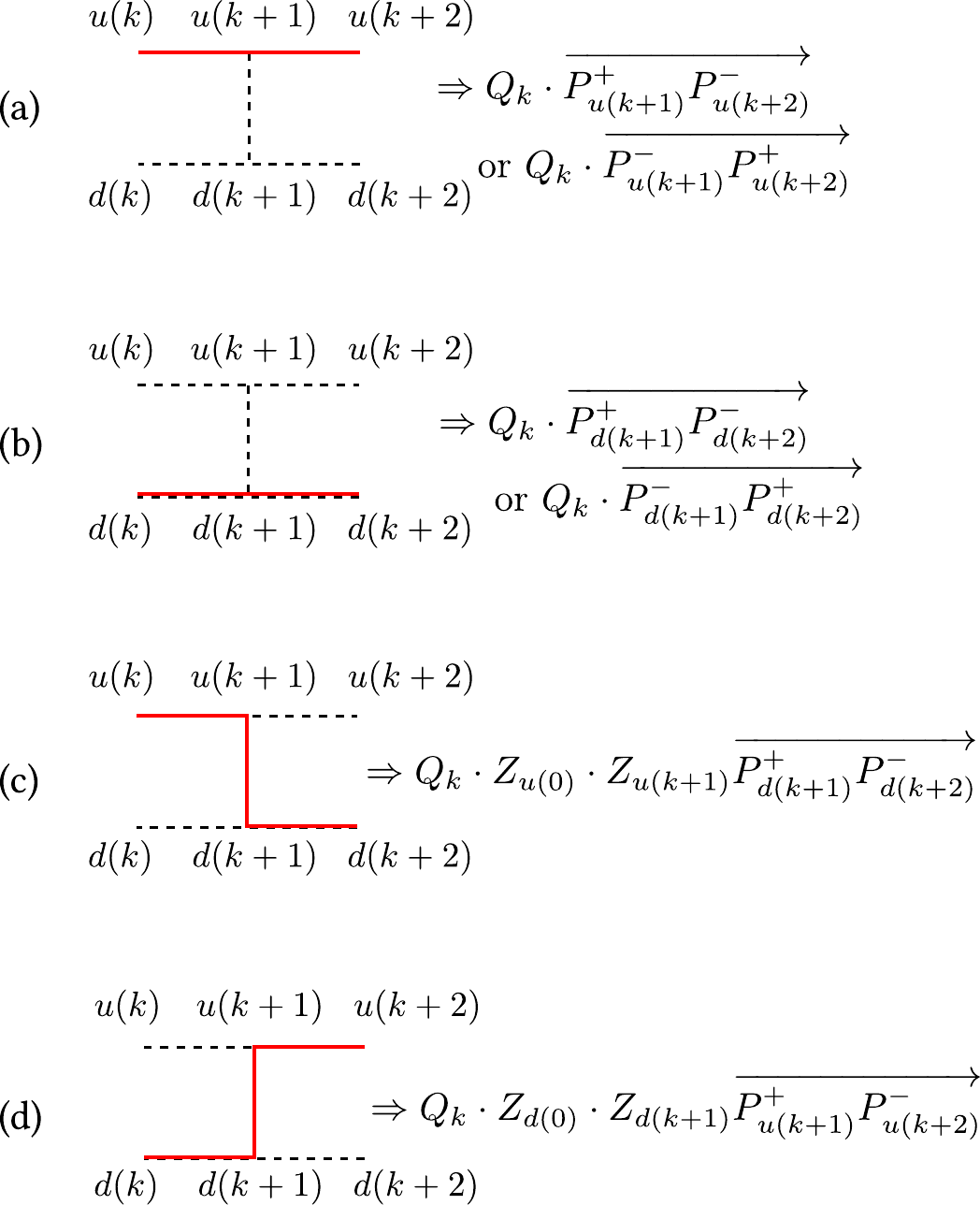}
\caption{\label{fig:operation_all} Schematic representation of the procedure of constructing the Pauli strings $Q$ from the paths.}
\end{figure}

\section{\label{sec:algebraicfastforwarding}Applying to Cartan-based fast-forwarding method}
Theorem~\ref{thm:hamalgdim} is an important result for discussing the scaling behavior of the dimension of the Hamiltonian algebra. However, we cannot apply this theorem directly to the discussion of the Cartan-based fast-forwarding method, one of the efficient implementation methods of time evolution operators. The detail of this method is provided in Ref.~\cite{kokcu2022cartanfastforward}, but we summarize the procedure of this method in the following for clarity. When the Hamiltonian algebra $\mathfrak{g}(H)$ is given, the procedure of the Cartan-based fast-forwarding method is summarized as follows:
\begin{enumerate}
    \item We decompose the Hamiltonian algebra $\mathfrak{g}(H)$ into two subspaces $\mathfrak{k}$ and $\mathfrak{m}$:
    \begin{align}
        \mathfrak{g}(H)=\mathfrak{k}\oplus\mathfrak{m}, 
    \end{align}
    where these subspaces $\mathfrak{k}$ and $\mathfrak{m}$ must satisfy the following conditions;
    \begin{align}
        \label{eq:kkink}[\mathfrak{k},\mathfrak{k}]&\subset\mathfrak{k},\\
        \label{eq:mmink}[\mathfrak{m},\mathfrak{m}]&\subset\mathfrak{k},\\
        \label{eq:mkinm}[\mathfrak{m},\mathfrak{k}]&\subset\mathfrak{m}.
    \end{align}
    In addition, $\mathfrak{k}$ and $\mathfrak{m}$ need to be constructed so that all of the Pauli strings appearing in the terms of the Hamiltonian $H$ belong to $\mathfrak{m}$. That is, if the Hamiltonian takes the form of $H=\sum_{i}\lambda_iP_i$ ($\lambda_i\in\mathbb{R}$), the subspace $\mathfrak{m}$ needs to satisfy
    \begin{align}
        \mathrm{Span}_{\mathbb{R}}(\{P_i\}_{i})\subset\mathfrak{m}.
    \end{align}
    The above decomposition can be realized using an involution $\theta:\mathfrak{g}(H)\to\mathfrak{g}(H)$. The involution $\theta$ is a Lie algebra homomorphism which satisfies $\theta^2=1$. Then, by using the involution $\theta$, the subspace $\mathfrak{k}$ ($\mathfrak{m}$) can be defined as the eigenspace of $\theta$ with eigenvalue $+1$ ($-1$).
    \item We generate the maximal abelian subalgebra of $\mathfrak{m}$ called the Cartan subalgebra. We denote it as $\mathfrak{h}$. The schematic picture of the relations among $\mathfrak{m}$, $\mathfrak{k}$, and $\mathfrak{h}$ is shown in Fig.~\ref{fig:cartanalgebra}.
    \item According to the "KHK" theorem~\cite{helgason2001differential,Earp2005CartanDecompositionSu2N,kokcu2022cartanfastforward},  for any $m\in\mathfrak{m}$, there exist $h \in \mathfrak{h}$ and $K=\prod_{j}e^{i\theta_jk_j}$ ($\{ k_j \}$ is the Pauli-string basis of $\mathfrak{k}$ and $\theta_j\in\mathbb{R}$) such that $m = KhK^\dag$. We apply this theorem to the qubit Hamiltonian $H\in\mathfrak{m}$ 
    to construct the time evolution operator,
    \begin{align}
        e^{iHt}=Ke^{ith}K^{\dagger}.
    \end{align}
\end{enumerate}
According to the above process, the circuit depth of the Cartan-based fast-forwarding method is determined by the dimension of $\mathfrak{k}$ and $\mathfrak{h}$. Hence, in general, there is a possibility that, although the dimension of the entire Hamiltonian algebra exhibits an exponential scaling, the dimension of the subspace $\mathfrak{k}\oplus\mathfrak{h}$ scales polynomially and we could implement the Cartan-based fast-forwarding method efficiently. However, as we show below, this is not the case for the Hamiltonian algebra $\mathfrak{g}(H_{1,\text{JW}})$. 

Hereafter in this section, $\mathfrak{k}$, $\mathfrak{m}$ and $\mathfrak{h}$ denote those corresponding to the case where the Lie algebra is $\mathfrak{g}(H_{1,\text{JW}})$.

\begin{thm}\label{thm:alg_m_exponential}
The subalgebra $\mathfrak{k}$ satisfies the following inequality:
\begin{align}\label{eq:dim_cartan_k_ineq}
    \text{dim}(\mathfrak{k})\ge 2^{N-3}.
\end{align}
\end{thm}
\begin{proof}
First, according to the proof of Lemma~\ref{lem:ladderparts}, $\overrightarrow{X_{u(0)}P^+_{u(1)}}$ and $\overrightarrow{Y_{u(0)}P^-_{u(1)}}$ are constructed by iteratively taking the commutators, the same number of times, of the Pauli strings appearing in the Hamiltonian, which belong to the subspace $\mathfrak{m}$. Hence, depending on the parity of the number of iterations, one of the following statements is true.
\begin{enumerate}
    \item Both $\overrightarrow{X_{u(0)}P^+_{u(1)}}$ and $\overrightarrow{Y_{u(0)}P^-_{u(1)}}$ belong to $\mathfrak{m}$.\\
    \item Both $\overrightarrow{X_{u(0)}P^+_{u(1)}}$ and $\overrightarrow{Y_{u(0)}P^-_{u(1)}}$ belong to $\mathfrak{k}$. 
\end{enumerate}    
In the case 2., because of  $Z_{u(0)}Z_{d(0)}\in\mathfrak{m}$ and Eq.~\eqref{eq:mkinm}, both $\overrightarrow{X_{u(0)}P^+_{u(1)}}\cdot Z_{u(0)}Z_{d(0)}$ and $\overrightarrow{Y_{u(0)}P^-_{u(1)}}\cdot Z_{u(0)}Z_{d(0)}$ belong to $\mathfrak{m}$. A similar argument holds for $\overrightarrow{X_{d(0)}P^+_{d(1)}}$ and $\overrightarrow{Y_{d(0)}P^-_{d(1)}}$. Hence, in the case of 1 (2), we define 
\begin{align}
    S_{X}^{u}=\overrightarrow{X_{u(0)}P^+_{u(1)}}\quad(\text{or } \overrightarrow{X_{u(0)}P^+_{u(1)}}\cdot Z_{u(0)}Z_{d(0)}),\\
    S_{Y}^{u}=\overrightarrow{Y_{u(0)}P^-_{u(1)}}\quad(\text{or }\overrightarrow{Y_{u(0)}P^-_{u(1)}}\cdot Z_{u(0)}Z_{d(0)}),
\end{align}
which satisfy $S_{X}^{u},S_{Y}^{u}\in\mathfrak{m}$. We also define $S_{X,Y}^d\in\mathfrak{m}$ in a similar manner. 

Then, we prepare the Pauli string $Q_{N-1}$ in $\mathfrak{g}(H)$ obtained by the procedure written in Sec.~\ref{subsec:proof interact}. By construction, $Q_{N-1}$ belongs to either $\mathfrak{k}$ or $\mathfrak{m}$. For simplicity, we suppose that the path corresponding to $Q_{N-1}$ starts from $(i_1,j_1)=(u(0),u(1))$. (When $(i_1,j_1)=(d(0),d(1))$, we replace $S_{X,Y}^u$ in the following discussion with $S_{X,Y}^d$.) If the string $Q_{N-1}$ belongs to $\mathfrak{k}$, the process is over. Then, we consider the case $Q_{N-1}\in\mathfrak{m}$. The product of $S_X^u$ and $S_Y^u$ is $S_X^uS_Y^u\propto Z_{u(0)}Z_{u(1)}$. In addition, the possible form of $Q_{N-1}$ is either
\begin{align}
    Q=X_{u(1)}W_{u(0),d(0)}\cdot (\cdots)
    \text{ or }Y_{u(1)}W_{u(0),d(0)}\cdot (\cdots)
\end{align}
where $W_{u(0),d(0)}=I$, $Z_{u(0)}$, $Z_{d(0)}$ or $Z_{u(0)}Z_{d(0)}$ and $(\cdots)$ consists of the Pauli matrices at the sites except for $u(1)$, $u(0)$, or $d(0)$. Hence, we can see that
\begin{align}
\{Q,S_X^uS_Y^u\}=0.
\end{align}
This shows that one of the Pauli string $S_X^u$ or $S_Y^u$ anticommutes with $Q_{N-1}$. Using $Q_{N-1}$, $S_X^u$, $S_Y^u\in\mathfrak{m}$ and Eq.~\eqref{eq:mmink}, the non-vanishing one of the commutators $[Q_{N-1}, S_X^u]$ or $[Q_{N-1}, S_Y^u]$ belongs to $\mathfrak{k}$.

Using the above procedure, we can construct at least $2^{N-3}$ elements of $\mathfrak{k}$, which are distinguished from each other by the positions $(u(i),d(i))$ ($i=2,3,\cdots,N-2$) of $X$ and $Y$ Pauli matrices. 
\end{proof}
Using Theorem~\ref{thm:alg_m_exponential}, we can say that
\begin{align}
    \text{dim}(\mathfrak{k}\oplus\mathfrak{h})\ge \text{dim}(\mathfrak{k})\ge 2^{N-3}.
\end{align}
In conclusion, the circuit depth for the Cartan-based fast-forwarding method of interacting fermion model scales exponentially with the number of sites.

\begin{figure}[t]
\includegraphics[scale=0.5]{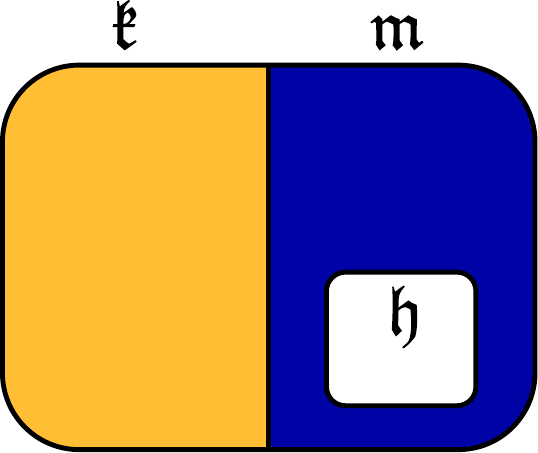}
\caption{\label{fig:cartanalgebra} The relation of the subalgebra $\mathfrak{k}$, $\mathfrak{m}$, and the Cartan subalgebra $\mathfrak{h}$ in $\mathfrak{g}(H)$.}
\end{figure}

To establish the exponential scaling of the circuit depth of the Cartan-based fast-forwarding method, there is an alternative proof independent of the argument given above. Once the inequality
\begin{gather}\label{eq:dimmdimhdimk_maintext}
	\text{dim}(\mathfrak{m})\le\text{dim}(\mathfrak{h})+\text{dim}(\mathfrak{k}),
\end{gather}
is proved, the argument proceeds straightforwardly. From Eq.~\eqref{eq:dimmdimhdimk_maintext}, we have
\begin{align}\label{eq:dimgdimh2dimk_maintext}
\text{dim}(\mathfrak{g}(H_{1,JW}))\le\text{dim}(\mathfrak{h})+2\text{dim}(\mathfrak{k}).
\end{align}
By Theorem.~\ref{thm:hamalgdim}, the left-hand side scales exponentially with the system size. Hence, at least one of $\text{dim}(\mathfrak{h})$ or $\text{dim}(\mathfrak{k})$ must exhibit exponential growth. Since the circuit depth of the Cartan-based fast-forwarding method is $\mathcal{O}(\text{dim}(\mathfrak{h}\oplus\mathfrak{k}))$, the depth also scales exponentially.

The proof of Eq.~\eqref{eq:dimmdimhdimk_maintext}, which relies on the restricted root decomposition and the Iwasawa decomposition, is given in Appendix.~\ref{appsec:restrictedroot}. We note that Eq.~\eqref{eq:dimmdimhdimk_maintext} and Eq.~\eqref{eq:dimgdimh2dimk_maintext} can also be derived from the surjectivity of the smooth maps $\Psi:K\times\mathfrak{h}\to\mathfrak{m}$ and $\Phi:K\times H\times K\to G$ ($K=\exp(\mathfrak{k})$, and $G=\exp(\mathfrak{g}(H))$) such that:
\begin{align}
\Psi(k,h)&=\text{Ad}_k(h),\\
\Phi(k_1,h,k_2)&=k_1hk_2.
\end{align}
In general, by Sard's theorem, the existence of a smooth surjective map $f:M\to N$ between two manifolds $M$ and $N$ implies $\text{dim}M\ge\text{dim}N$ \cite{lee2012smooth}.

\section{\label{sec:generalmappings}Generaliation to other fermion-qubit mappings}

In the previous sections, we focused on the Jordan-Wigner-transformed fermion models. In this section, we extend it to the general fermion-qubit mappings. We introduce the following Majorana fermions
\begin{gather}
	\gamma_{+,i\sigma}=c^{\dagger}_{i\sigma}+c_{i\sigma},\\
	\gamma_{-,i\sigma}=-i(c_{i\sigma}-c^{\dagger}_{i\sigma}).
\end{gather}
We can write the Hamiltonian Eq.~\eqref{eq:oneintfermion} in terms of the Majorana fermions

\begin{align}\label{eq:oneinthamil_majorana}
	H_1=&\sum_{\sigma=\uparrow,\downarrow}\sum_{(i,j)\in G}\frac{t_{ij}}{2}
	(i\gamma_{-,i\sigma}\gamma_{+,j\sigma}-i\gamma_{+,i\sigma}\gamma_{-,j\sigma})\nonumber\\
	&+U\left(i\gamma_{+,0\uparrow}\gamma_{-,0\uparrow}\right)\left(i\gamma_{+,0\downarrow}\gamma_{-,0\downarrow}\right).
\end{align}
If all of the products of the majorana operators $\gamma_{+,i\sigma}\gamma_{-,j\sigma}$ is mapped to single-term Pauli strings (not the linear combination of plural Pauli strings) by a given injective fermion-qubit mapping, the Hamiltonian algebra for the qubit-mapped fermion models is isomorphic to the Lie algebra generated by the following operators:
\begin{gather}
	\gamma_{+,i\sigma}\gamma_{-,j\sigma}, \gamma_{-,i\sigma}\gamma_{+,j\sigma}\quad((i,j)\in G),\label{eq:maj_hoppingterm}\\
	\left(\gamma_{+,0\uparrow}\gamma_{-,0\uparrow}\right)\left(\gamma_{+,0\downarrow}\gamma_{-,0\downarrow}\right)\label{eq:maj_int}.
\end{gather}
Similar to Theorem.~\ref{thm:hamalgdim}, the dimension of this Hamiltonian algebra satisfies the following relation.

\begin{thm}\label{thm:int_generalhamil}
	We define the Lie algebra $\mathfrak{g}(H_1,\gamma)$ as the Lie algebra generated by the operators of Majorana fermions Eq.~\eqref{eq:maj_hoppingterm} and Eq.~\eqref{eq:maj_int}. Then, the dimension of the Lie algebra $\mathfrak{g}(H_1,\gamma)$ satisfies
	\begin{align}\label{eq:dim_ineq_hamilang_interact_general}
		\text{dim}(\mathfrak{g}(H_1,\gamma))\ge 2^{N-1},
	\end{align}
	and it scales exponentially with $N$. Therefore, under the fermion-qubit mapping that maps all of the products of the Majorana operators to the single-term Pauli strings, the Hamiltonian algebra for the qubit-mapped models of the Hamiltonian Eq.~\eqref{eq:oneintfermion} (we write it as $\mathfrak{g}(H_1)$ for simplicity) satisfies $\mathfrak{g}(H_1)\ge 2^{N-1}$. 
\end{thm}
The proof can be obtained by directly rewriting the proof of Theorem~\ref{thm:hamalgdim} in the Majorana-fermion representations. We give the proof in Appendix.~\ref{appsec:proof_of_int_hamil_general}. Similarly, we obtain the following theorem, which corresponds to Theorem.~\ref{thm:alg_m_exponential}.
\begin{thm}\label{thm:alg_m_exponential_general}
	The subalgebra $\mathfrak{k}$ of the Hamiltonian algebra $\mathfrak{g}(H_1,\gamma)$ (or $\mathfrak{g}(H_1)$) satisfies \begin{align}\label{eq:dim_k_general}
    \text{dim}(\mathfrak{k})\ge 2^{N-3}. 
    \end{align}
\end{thm}

To apply these theorems to the qubit mapped models, a fermion-qubit mapping needs to map all of the products of the Majorana operators into single-term Pauli strings. Here we list the explicit fermion-qubit mappings subject to this condition. As discussed below, this framework encompasses many standard fermion-to-qubit mappings introduced in the literature. Below, for simplicity, we write $i$ for the pair $(i,\sigma)$.

\subsection{Jordan-Wigner transformation}
The Jordan-Wigner transformation is defined in Eq.~\eqref{eq:jw_fermion_def_dag} and \eqref{eq:jw_fermion_def}, which are rewritten in the form

\begin{align}
	\gamma_{+,j}=Z_1Z_2\cdots Z_{j-1}X_{j},\\
	\gamma_{-,j}=Z_1Z_2\cdots Z_{j-1}Y_{j}.
\end{align}
These transformations imply that the product of the Majorana operators $\gamma_{+,i\sigma}\gamma_{-,j\sigma}$ is mapped to a single-term Pauli string. Therefore, the Jordan-Wigner transformation satisfies the condition for fermion-qubit mappings stated in Theorem.~\ref{thm:int_generalhamil}.

\subsection{Bravyi-Kitaev transformation}
The Bravyi-Kitaev transformation~\cite{Bravyi2002210_BK,Seely2012BK_parity} is defined in the following form:
\begin{align}
	\gamma_{+,j}=\left(\bigotimes_{k\in U(j)}X_{k}\right) X_{j} \left(\bigotimes_{k\in P(j)}Z_{k}\right),\\
	\gamma_{-,j}=\left(\bigotimes_{k\in U(j)}X_{k}\right) Y_{j} \left(\bigotimes_{k\in P(j)}Z_{k}\right).
\end{align}
$U(j)$ and $P(j)$ are certain kinds of subsets of indices associated with $j$. We follow the notation of Ref~\cite{Seely2012BK_parity}. Bravyi-Kitaev transformation maps both the Majorana operators $\gamma_{+,j}$ and $\gamma_{-,j}$ to single-term Pauli strings, hence this also maps $\gamma_{+,i\sigma}\gamma_{-,j\sigma}$ to a single-term Pauli string.

The variants of the Bravyi-Kitaev transformation, such as the  Bravyi-Kitaev superfast encoding~\cite{Setia2018_BKSuperFast}, have also been proposed. These variants also satisfy the condition that each Majorana operators is mapped to a single-term Pauli string.

\subsection{Parity transformation}
The parity transformation~\cite{Seely2012BK_parity} is defined in the following form:
\begin{align}
	\gamma_{+,j}&=Z_{j-1}X_{j}\bigotimes_{k>j}Z_k,\\
	\gamma_{-,j}&=Y_{j}\bigotimes_{k>j}Z_k.
\end{align}
Hence, the parity transformation maps the product $\gamma_{+,i\sigma}\gamma_{-,j\sigma}$ to a single-term Pauli string.

\subsection{Compact fermion mapping}
The compact fermion mapping~\cite{Derby2021CptMapping} is defined as
\begin{gather}
	\tilde{E}_{jk}=-i\gamma_{+,j}\gamma_{+,k}=\pm X_{j}Y_{k}X_{f(j,k)},\\
	\tilde{V}_{j}=-i\gamma_{+,j}\gamma_{-,j}=Z_j.
\end{gather}
$f(j,k)$ denotes a certain site determined by $j$ and $k$. The sign of $\tilde{E}_{jk}$ is determined by the positional relationship between sites $i$ and $j$. Since the operators $\tilde{E}_{jk}$ and $\tilde{V}_{j}$ are represented by single-term Pauli strings and the product $\gamma_{+,j}\gamma_{-,j}$ is proportional to $\tilde{E}_{jk}\tilde{V}_{j}$, the compact fermion mapping satisfies the condition for a fermion-qubit mapping mentioned in Theorem.~\ref{thm:int_generalhamil}.

\subsection{Other fermion-qubit mappings}
There are, in principle, infinitely many possible fermion-qubit mappings, so it is not feasible to list them all. Other representative examples include the ternary-tree mappings and its variants~\cite{Jiang2020optimalfermiontoTernary,Miller2023TernaryBonsai,Liu2025HamilAdaptiveTernary}. These mappings satisfy the condition that each Majorana operator is mapped to a single-term Pauli string.

\section{Examples}\label{sec:application}
We apply Theorem~\ref{thm:hamalgdim} and Theorem~\ref{thm:alg_m_exponential} (or Theorem~\ref{thm:int_generalhamil} and Theorem~\ref{thm:alg_m_exponential_general}) to the Hamiltonian algebra of some prototypical fermion models. In the previous study, the dimensions of the Hamiltonian algebras of the following fermion models are empirically shown to scale exponentially with the number of qubits ~\cite{Steckmann2023AIMHubbard}. In this paper, we prove that this holds rigorously.

In this section, we restrict our discussion to the Jordan-Wigner transformation, but the same argument applies to the other fermion-qubit mappings listed in Sec.~\ref{sec:generalmappings}.

\subsection{Anderson model}\label{subsec:andersonmodel}
First, we consider the Anderson impurity model (AIM). This model describes metals with a magnetic impurity. We assume that, for simplicity, there are no hopping terms between the itinerant electron sites. The Hamiltonian of the AIM is given by
\begin{align}\label{eq:aim}
    H_{\text{AIM}}=&\sum_{\sigma}\sum_{i=1}^{N-1}V_{i}(c^{\dagger}_{i\sigma}c_{0\sigma}+\text{h.c.})\nonumber\\
    &+U\left(n_{0\uparrow}-\frac12\right)\left(n_{0\downarrow}-\frac12\right).
\end{align}
In AIM, site $0$ corresponds to the impurity site and the other sites correspond to the itinerant electron sites. The interaction between spin-up and spin-down electrons occurs only at the impurity site. The corresponding graph representing hopping terms is shown in Fig.~\ref{fig:graph_anderson} and is a connected graph. By performing the Jordan-Wigner transformation, the AIM is transformed to the qubit Hamiltonian
\begin{align}\label{eq:aim_jw}
    H_{\text{AIM,JW}}=&\sum_{i=1}^{N-1}\frac{V_{i}}{2}(\overrightarrow{X_0X_i}+\overrightarrow{Y_0Y_i})\nonumber\\
    &+\sum_{i=1}^{N-1}\frac{V_{i}}{2}(\overrightarrow{X_NX_{N+i}}+\overrightarrow{Y_NY_{N+i}})\nonumber\\
    &+\frac{U}{4}Z_0Z_N,
\end{align}
where the indices $0$ to $N-1$ correspond to the spin-up sites and $N$ to $2N-1$ correspond to the spin-down sites. According to Theorem.~\ref{thm:hamalgdim}~(2), the Hamiltonian algebra of the AIM satisfies
\begin{align}
    \text{dim}(\mathfrak{g}(H_{\text{AIM,JW}}))\ge 2^{N-1}.
\end{align}
We may consider the AIM without the constant $\frac12$, which seems to be more conventional,
\begin{align}\label{eq:conventionalAIM}
H_{\text{convAIM}}=&\sum_{\sigma}\sum_{i=1}^{N-1}V_{i}(c^{\dagger}_{i\sigma}c_{0\sigma}+\text{h.c.})+Un_{0\uparrow}n_{0\downarrow}.
\end{align}
However, the JW transformation of the Hamiltonian Eq.~\eqref{eq:conventionalAIM} is 
\begin{align}
H_{\text{convAIM,JW}}=H_{\text{AIM,JW}}+\frac{U}{8}Z_{0}+\frac{U}{8}Z_{N}+\frac{U}{16}.
\end{align}
The Hamiltonian algebra for this Hamiltonian is generated from the Pauli strings appearing in the terms of $H_{\text{AIM,JW}}$ and the additional Pauli strings ($Z_{0}$ and $Z_N$). Thus, according to Proposition~\ref{prop:generatedliealg}, the dimension of the Hamiltonian algebra for the Hamiltonian $H_{\text{convAIM,JW}}$ satisfies $\text{dim}(\mathfrak{g}(H_{\text{convAIM,JW}}))\ge\text{dim}(\mathfrak{g}(H_{\text{AIM,JW}}))$. Since $\text{dim}(\mathfrak{g}(H_{\text{AIM,JW}}))$ grows exponentially, $\text{dim}(\mathfrak{g}(H_{\text{convAIM,JW}}))$ also grows exponentially.

We may adopt other ways of assigning qubit indices, such as one in which even numbers $\{0,2,4,\cdots,2N-2\}$ correspond to the spin-up sites and odd numbers $\{1,3,5,\cdots,2N-1\}$ correspond to the spin-down sites. However, regardless of the choice of indexing, Theorem.~\ref{thm:hamalgdim}~(2) guarantees the exponential dimensionality of the Hamiltonian algebra of the AIM. In addition, we may consider additional chemical potential or other terms. These terms introduce additional Pauli strings to the Hamiltonian. Hence the dimension of the resulting Hamiltonian algebra is larger than that of $H_{\text{AIM,JW}}$ by Prop.~\ref{prop:generatedliealg}. Therefore, the Hamiltonian algebra of the AIM with additional terms still exhibits an exponential scaling of the dimension. Similarly, the circuit depth of the Cartan-based fast-forwarding method scales exponentially with the number of sites according to Theorem~\ref{thm:alg_m_exponential}.

\begin{figure}[t]
\includegraphics[scale=0.5]{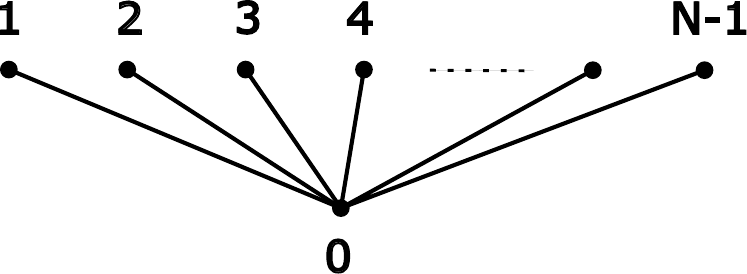}
\caption{\label{fig:graph_anderson} Graph representing hopping terms appearing in the Anderson impurity model. This graph is connected, since all two vertices are connected via the site $0$.}
\end{figure}

\subsection{Hubbard model}\label{subsec:hubbardmodel}
Next, we consider the Hubbard model:
\begin{align}\label{eq:hubbard}
    H_{\text{Hubbard}}=&-t\sum_{\sigma}\sum_{(ij)\in G}(c^{\dagger}_{i\sigma}c_{j\sigma}+\text{h.c.})\nonumber\\&+U\sum_{i}\left(n_{i\uparrow}-\frac12\right)\left(n_{i\downarrow}-\frac12\right),
\end{align}
where $G$ denotes a general connected graph, such as a one-dimensional chain or two-dimensional square lattice. By performing the Jordan-Wigner transformation, Eq.~\eqref{eq:hubbard} becomes
\begin{align}\label{eq:hubbard_jw}
    H_{\text{Hubbard}}=
    &-\frac{t}{2}\sum_{\left<ij\right>\in G}(\overrightarrow{X_iX_j}+\overrightarrow{Y_iY_j})\nonumber\\
    &-\frac{t}{2}\sum_{\left<ij\right>\in G}(\overrightarrow{X_{N+i}X_{N+j}}+\overrightarrow{Y_{N+i}Y_{N+j}})\nonumber\\
    &+4U\sum_{i}Z_{i}Z_{i+N}.
\end{align}
In this representation, we assign the spin-up sites to the zeroth to $(N-1)$-th qubits, and the spin-down sites to the other qubits. In order to apply Theorem~\ref{thm:hamalgdim} to the Hubbard model, we consider the following Jordan-Wigner transformed fermion model with a single interaction term,
\begin{align}
    H_{\text{Hubbard,JW,1}}=&-\frac{t}{2}\sum_{\left<ij\right>\in G}(\overrightarrow{X_iX_j}+\overrightarrow{Y_iY_j})\nonumber\\
    &-\frac{t}{2}\sum_{\left<ij\right>\in G}(\overrightarrow{X_{N+i}X_{N+j}}+\overrightarrow{Y_{N+i}Y_{N+j}})\nonumber\\
    &+4UZ_{0}Z_{N}.
\end{align}
According to Prop.~\ref{prop:generatedliealg}, 
the Hamiltonian algebras of these models satisfy 
\begin{align}
    \text{dim}(\mathfrak{g}(H_{\text{Hubbard,JW}}))\ge\text{dim}(\mathfrak{g}(H_{\text{Hubbard,JW,1}})).
\end{align}
On the other hand, by using Theorem~\ref{thm:hamalgdim}, inequality $\text{dim}(\mathfrak{g}(H_{\text{Hubbard,JW,1}}))\ge 2^{N-1}$ holds, and hence $\text{dim}(\mathfrak{g}(H_{\text{Hubbard,JW}}))\ge 2^{N-1}$ holds. Therefore, the dimension of the Hamiltonian algebra of the Hubbard model scales exponentially, and so does the circuit depth for the Cartan-based fast-forwarding method.

\section{Exceptions}\label{sec:exceptions}
Although we have considered the property of the Hamiltonian algebra of the wide class of fermion models, we cannot apply Theorem~\ref{thm:hamalgdim} to several classes of fermion models. First, in quantum chemistry, we often consider fermion models such as
\begin{align}
    H=\sum_{\sigma}\sum_i\epsilon_ic^{\dagger}_{i\sigma}c_{i\sigma}+\sum_{\sigma,\sigma'}\sum_{pqrs}V_{pqrs}c^{\dagger}_{p\sigma}c^{\dagger}_{q\sigma'}c_{r\sigma'}c_{s\sigma},
\end{align}
where $c_{i\sigma}$ and $c^{\dagger}_{i\sigma}$ are the annihilation and creation operators of a spin-orbital $(i\sigma)$. In this case, the graph representing the quadratic hopping terms does not need to be connected. Hence, although there is a possibility that the fourth-order interaction may push the dimension of the Hamiltonian algebra to the exponential order of the number of sites in this model, we cannot discuss this possibility by utilizing Theorem~\ref{thm:hamalgdim}.

Next, the construction of the exponential number of strings in $\mathfrak{g}(H_{1,\text{JW}})$ owes to Lemma~\ref{lem:ladderparts}. If one of the strings $\overrightarrow{X_{i}X_{j}}$ and $\overrightarrow{Y_{i}Y_{j}}$ is missing in the Jordan-Wigner transformed qubit model, the Hamiltonian becomes out of the range of Theorem~\ref{thm:hamalgdim}. This case occurs, for example, when the Hamiltonian takes the following form:
\begin{align}
    H=&\sum_{\sigma}\sum_{\left<ij\right>\in G,i<j}(c^{\dagger}_{i\sigma}-c_{i\sigma})(c^{\dagger}_{j\sigma}+c_{j\sigma})\nonumber\\ &+U\left(n_{0\uparrow}-\frac12\right)\left(n_{0\downarrow}-\frac12\right).
\end{align}
In particular, if the graph representing hopping terms is a one-dimensional chain, the Jordan-Wigner transformed Hamiltonian becomes
\begin{align}\label{eq:onlyx}
    H_{X}=&\sum_{\sigma}\sum_{i=0}^{N-2}(X_{i}X_{i+1}+X_{N+i}X_{N+i+1})+4UZ_0Z_N.
\end{align}
Because the Pauli strings $X_{1}X_{2}$, $X_{2}X_{3}$, $\cdots$, $X_{N-2}X_{N-1}$ and $X_{N+1}X_{N+2}$, $X_{N+2}X_{N+3}$, $\cdots$, $X_{2N-2}X_{2N-1}$ commute with all of the terms in the Hamiltonian, the Hamiltonian algebra of this model is 
\begin{align}
    \mathfrak{g}(H_X)=&\text{Span}\left(\right.\nonumber\\
    &\{X_0X_1,X_NX_{N+1},Z_0Z_N,Y_0X_1Z_N,Z_0Y_NX_{N+1}\}\nonumber\\
    &\cup\{X_{i}X_{i+1}|i=1,2,\cdots,N-2\}\nonumber\\
    &\left.\cup\{X_{N+i}X_{N+i+1}|i=1,2,\cdots,N-2\}\right).
\end{align}
Hence the dimension of the Hamiltonian in Eq.~\eqref{eq:onlyx} is $\text{dim}(\mathfrak{g}(H_{X}))=2N+1$ and is polynomial order of the number of sites, despite the presence of interaction term between spin-up and spin-down sites. 

Finally, since our argument for constructing the basis of the Hamiltonian algebra is based on the spin degree of freedom, we cannot apply Theorem~\ref{thm:hamalgdim} to spinless fermion models.

\section{Conclusion and Discussion}\label{sec:conclusionanddiscussion}
We investigate the Hamiltonian algebra of qubit-mapped fermion models. First, we consider the qubit-mapped free fermion model. The Hamiltonian algebra of this model exhibits a quadratic-order scaling of the number of sites. Then we consider the fermion model with a single-site Coulomb interaction. We show that the dimension of the Hamiltonian algebra of this model exhibits an exponential scaling of the number of sites. By applying this theorem, we show that the dimension of the Hamiltonian algebra of the AIM and the Hubbard model exhibits an exponential scaling. This result implies the hardness of the Cartan-based fast-forwarding method efficiently with these models due to the exponential scaling of circuit depth.

This paper evaluates the inequalities for the dimension of the Hamiltonian algebra, focusing on its order, which is crucial for quantum algorithm implementation. The obtained bounds are not tight, but they are sufficient for our purposes. 
Specifically, the inequalities Eq.~\eqref{eq:dim_ineq_hamilang_interact} and Eq.~\eqref{eq:dim_cartan_k_ineq} (Eq.~\eqref{eq:dim_ineq_hamilang_interact_general} and Eq.~\eqref{eq:dim_k_general}) are not tight evaluations. Thus, future work will focus on obtaining tighter bounds and elucidating the explicit structure of the Hamiltonian algebras for specific fermion models with interactions.

In our study, we have considered qubit models derived from certain interacting fermion models, including AIM and the Hubbard model, via fermion–qubit mappings that map Majorana operators to single-term Pauli strings, and found that the circuit depth of the Cartan-based fast-forwarding method scales exponentially with the number of qubits. However, the fast-forwardability of these models is not completely prohibited. Thus, the following two things are worth discussing. 

First, a choice of fermion-qubit mappings may alter the algebraic structure of the Hamiltonian algebra. For example, we can consider the interacting fermion model without the hopping terms involving the interacting sites, such as
\begin{align}\label{eq:model_t_U_decouple}
	H=-t\sum_{\sigma}\sum_{i\ge 2}(c^{\dagger}_{i\sigma}c_{i+1\sigma}+\text{h.c.})+U\left(n_{1\uparrow}-\frac12\right)\left(n_{1\downarrow}-\frac12\right).
\end{align}
The dimension of its Hamiltonian algebra scales quadratically with the number of sites when we perform the Jordan-Wigner transformation defined in Eqs.~\eqref{eq:jw_fermion_def_dag} and ~\eqref{eq:jw_fermion_def}. On the other hand, when we perform the following fermion-qubit mapping,
\begin{align}\label{eq:orbitalrotation12_1}
	c_{1\sigma}=\frac{1}{\sqrt2}\left(Z_1\cdots Z_{\sigma(1)-1}\frac{X_{\sigma(1)}-iY_{\sigma(1)}}{2}\right.\nonumber\\\left.-Z_1\cdots Z_{\sigma(2)-1}\frac{X_{\sigma(2)}-iY_{\sigma(2)}}{2}\right),\\
    \label{eq:orbitalrotation12_2}
	c_{2\sigma}=\frac{1}{\sqrt2}\left(Z_1\cdots Z_{\sigma(1)-1}\frac{X_{\sigma(1)}-iY_{\sigma(1)}}{2}\right.\nonumber\\\left.+Z_1\cdots Z_{\sigma(2)-1}\frac{X_{\sigma(2)}-iY_{\sigma(2)}}{2}\right),\\\label{eq:orbitalrotation12_3}
	c_{k\sigma}=Z_1\cdots Z_{\sigma(k)-1}\frac{X_{\sigma(k)}-iY_{\sigma(k)}}{2}\quad (k\ge3),
\end{align}
where $\sigma(k)=u(k)$ or $d(k)$ depending on the spin $\sigma=\uparrow, \downarrow$, (which is the combination of the orbital rotation between the sites $1$ and 2 and the Jordan-Wigner transformation,) the dimension of the resulting Hamiltonian algebra scales exponentially (see Appendix.~\ref{appsec:orbitferionqubitmap}). Conversely, it is not prohibited that the prototypical interacting fermion models can be transformed so that the dimension of the resulting Hamiltonian algebras scales polynomially. It is challenging to prove the impossibility of the Cartan-based fast-forwarding method for interacting fermion models when using any other fermion-qubit mapping. 

Second, it should be noted that the computational cost of the fast-forwarding method under a finite error tolerance is beyond the scope of this study. When the computational accuracy is fixed to a certain finite value, there may be a possibility of efficiently performing the fast-forwarding method for interacting fermion models. Moreover, practical and useful computations might be achievable using methods such as the variational fast-forwarding method motivated by the Cartan-based fast-forwarding method. Investigating the fast-forwardability of interacting fermion models in terms of choices of fermion-qubit mappings and approximations of the Cartan-based fast-forwarding method remains a task for future work.

\begin{acknowledgments}
We are grateful to Yuya O. Nakagawa, Keita Kanno, Masaya Kohda, Hokuto Iwakiri, and Naohisa Sueishi for valuable discussion and advice.
\end{acknowledgments}

\appendix

\section{\label{app:correction}Modification of proof}
In Subsec.~\ref{subsec:proof interact} in the main text, we discuss the properties of the Hamiltonian algebra of Jordan-Wigner transformed interacting fermion model in Eq.~\eqref{eq:dim_ineq_hamilang_interact} restricting to the case $u(0)<u(i)$ and $d(0)<d(i)$ for any value $i=1,2,\cdots, N-1$. Then, we generalize the model to arbitrary choices of $u(0)$ and $d(0)$. However, the proof in the main text cannot be applied to the general cases as they are and we need to modify the discussion slightly.

In general cases, the proof of Lemma~\ref{lem:ladderparts} is almost the same as the special case except for replacing $u(0)$ ($d(0)$) with $u_{\text{min}}$ ($d_{\text{min}}$), where $u_{\text{min}}$ ($d_{\text{min}}$) is the minimum of the spin-up (spin-down) sites including the interacting site $u(0)$ ($d(0)$). If $u_{\text{min}}\neq u(0)$, by taking the commutation relations of $\overrightarrow{X_{u_{\text{min}}}P^+_{u(0)}}$, $\overrightarrow{Y_{u_{\text{min}}}P^-_{u(0)}}$, $\overrightarrow{X_{u_{\text{min}}}P^+_{u(i)}}$, and $\overrightarrow{X_{u_{\text{min}}}P^-_{u(i)}}$, we obtain 
\begin{align}
    \overrightarrow{P^-_{u(0)}P^+_{u(i)}},\quad \overrightarrow{P^+_{u(0)}P^-_{u(i)}}\in\mathfrak{g}(H_{1,\text{JW}}).
\end{align}
When $P^-_{u(0)}=Y_{u(0)}$, we swap the definition of $P^+_{u(i)}$ between $X_{u(i)}$ and $Y_{u(i)}$ for any $i=1,2,\cdots,N-1$. A similar discussion can be applied to the spin-down case.

The latter part of Lemma~\ref{lem:ladderparts} also holds for general cases. However, we have to take the commutators of different strings to obtain $\overrightarrow{P^+_{u(i)}P^-_{u(i+1)}}$ and $\overrightarrow{P^-_{u(i)}P^+_{u(i+1)}}$ depending on the order relations among $u(0)$, $u(i)$, and $u(i+1)$. 

Lemma~\ref{lem:rungparts} also needs a modification, because the proof in the main text uses the fact that $u(0)$ ($d(0)$) is the smallest value in $\{u(0),u(1),\cdots,u(N-1))\}$ ($\{d(0),d(1),\cdots,d(N-1))\}$). To generalize the proof of Lemma~\ref{lem:rungparts} to arbitrary $u(0)$ and $d(0)$, we introduce the following notations: 
\begin{align}
    \label{app_eq:Lxup}
    L_{X}^{\uparrow}(i)=\begin{cases}
        \overrightarrow{X_{u(0)}P^+_{u(i)}}\quad (u(i)>u(0))\\
        \overrightarrow{Y_{u(0)}P^-_{u(i)}}\quad (u(i)<u(0))
    \end{cases},\\
    \label{app_eq:Lyup}
    L_{Y}^{\uparrow}(i)=\begin{cases}
        \overrightarrow{Y_{u(0)}P^-_{u(i)}}\quad (u(i)>u(0))\\
        \overrightarrow{X_{u(0)}P^+_{u(i)}}\quad (u(i)<u(0))
    \end{cases},\\ 
    \label{app_eq:Lxdown}
    L_{X}^{\downarrow}(i)=\begin{cases}
        \overrightarrow{X_{d(0)}P^+_{d(i)}}\quad (d(i)>d(0))\\
        \overrightarrow{Y_{d(0)}P^-_{d(i)}}\quad (d(i)<d(0))
    \end{cases},\\
    \label{app_eq:Lydown}
    L_{Y}^{\downarrow}(i)=\begin{cases}
        \overrightarrow{Y_{d(0)}P^-_{d(i)}}\quad (d(i)>d(0))\\
        \overrightarrow{X_{d(0)}P^+_{d(i)}}\quad (d(i)<d(0))
    \end{cases}.
\end{align}
The important properties of the strings in Eq.~\eqref{app_eq:Lxup} to Eq.~\eqref{app_eq:Lydown} is that only the following ones are the nonvanishing commutators among them:
\begin{gather}
    [L_{X}^{\uparrow}(i),L_{X}^{\uparrow}(j)]\propto \overrightarrow{P^-_{u(i)}P^+_{u(j)}}\quad (i<j),\\
    [L_{Y}^{\uparrow}(i),L_{Y}^{\uparrow}(j)]\propto \overrightarrow{P^+_{u(i)}P^-_{u(j)}}\quad (i<j).
\end{gather}
and their counterparts for $L_{X,Y}^{\downarrow}$.  

Instead of $\Sigma$ strings in Eq.~\eqref{eq:Sigmas1} to Eq.~\eqref{eq:Sigmas4}, we define $\Sigma_i(k)$ ($i=1,2,3,4$) as
\begin{align}
    \Sigma_1(i)&=L_{X}^{\uparrow}(i)\cdot L_{X}^{\downarrow}(i)\cdot Z_{u(0)}Z_{d(0)},\\
    \Sigma_2(i)&=L_{Y}^{\uparrow}(i)\cdot L_{X}^{\downarrow}(i)\cdot Z_{u(0)}Z_{d(0)},\\
    \Sigma_3(i)&=L_{X}^{\uparrow}(i)\cdot L_{Y}^{\downarrow}(i)\cdot Z_{u(0)}Z_{d(0)},\\
    \Sigma_4(i)&=L_{Y}^{\uparrow}(i)\cdot L_{Y}^{\downarrow}(i)\cdot Z_{u(0)}Z_{d(0)}.
\end{align}
We can show that these strings belong to the Hamiltonian algebra $\mathfrak{g}(H_{1,\text{JW}})$. These definitions are the extensions of the $\Sigma_{1,2,3,4}$ in the main text. Then, by taking the commutator of these $\Sigma$ strings, as in the main text, we obtain the Pauli strings equivalent to Eq.~\eqref{eq:ZPPPP1} to Eq.~\eqref{eq:ZPPPP4}. Therefore, Lemma~\ref{lem:rungparts} also holds for the general $u(0)$ and $d(0)$.

\section{\label{appsec:restrictedroot}Restricted root decomposition and Iwasawa decomposition}
In this section, we prove the inequality for the Cartan decomposition
\begin{align}\label{eq:dimmdimhdimk}
	\text{dim}(\mathfrak{m})\le\text{dim}(\mathfrak{h})+\text{dim}(\mathfrak{k}).
\end{align}
The discussion below is based on Ref.~\cite{Helgason2001DifferentialGeometry,Knapp2002LieGroup}. We assume that the Cartan subalgebra (maximal abelian subalgebra of $\mathfrak{m}$) is spanned by elements $h_i\in\mathfrak{h}$ for $i=1,2,\cdots,\text{dim}\mathfrak{h}$. We choose this basis such that each $h_i$ are Pauli strings. For $\lambda=(\lambda_1,\cdots,\lambda_{\text{dim}\mathfrak{h}})$, we define the eigenspace $\mathfrak{g}_{\lambda}$ of the adjoint action as
\begin{align}
	\mathfrak{g}_{\lambda}=\{g\in\mathfrak{g}|\forall i,[h_i,g]=\lambda_i g\}=\{g\in\mathfrak{g}|\forall i,Ad_{h_i}(g)=\lambda_i g\}.
\end{align}
If $g\in\mathfrak{g}_{\lambda}$, then $g$ satisfies $[h_i,g]=\lambda_i g$ for all $i$. Applying the involution $\theta$ and using the fact that $\mathfrak{h}\subset\mathfrak{m}$, we obtain 
\begin{align}
	[h_i,\theta(g)]=-\lambda_i\theta(g).
\end{align}
Hence $\theta(\mathfrak{g}_{\lambda})=\mathfrak{g}_{-\lambda}$ so $\theta$ defines a vector-space isomorphism from $\mathfrak{g}_{\lambda}$ onto $\mathfrak{g}_{-\lambda}$.

We have the restricted root decomposition of the Lie algebra $\mathfrak{g}$:
\begin{align}
	\mathfrak{g}=\mathfrak{g}_0\oplus\bigoplus_{\lambda\neq 0}\mathfrak{g}_{\lambda}=\mathfrak{g}_0\oplus\bigoplus_{\lambda> 0}\mathfrak{g}_{\lambda}\oplus\bigoplus_{\lambda< 0}\mathfrak{g}_{\lambda}
\end{align}
where we define the order relation $\lambda>\lambda'$ as the lexicographical order. The subalgebra $\mathfrak{g}_0$ consists of the elements of $\mathfrak{g}$ that commute with all elements of the Cartan subalgebra $\mathfrak{h}$. We decompose $\mathfrak{g}_0$ into the part lying in $\mathfrak{m}$ and the part lying in $\mathfrak{k}$. The former is, by definition, the Cartan subalgebra $\mathfrak{h}$. We denote the latter by $Z_{\mathfrak{k}}(\mathfrak{h})$, i.e.
\begin{align}
	Z_{\mathfrak{k}}(\mathfrak{h}):=\{g\in\mathfrak{k}|\forall h\in\mathfrak{h},[h,g]=0\},
\end{align}
which is the centralizer of $\mathfrak{h}$ in $\mathfrak{k}$. Hence we obtain the decomposition
\begin{align}
	\mathfrak{g}_0=\mathfrak{h}\oplus Z_{\mathfrak{k}}(\mathfrak{h}).
\end{align} 

Since $\theta(\mathfrak{g}_{\lambda})=\mathfrak{g}_{-\lambda}$ for every $\lambda$, the involution $\theta$ induces an automorphism of the vector space $\mathfrak{g}_{\lambda}\oplus\mathfrak{g}_{-\lambda}$. As $\theta$ has eigenvalues $\pm1$, we obtain the corresponding eigendecomposition. We define the $\pm1$ eigenspaces by
\begin{align}
	m^{\pm}_{\lambda}:=\{g\in\mathfrak{g}_{\lambda}\oplus\mathfrak{g}_{-\lambda}|\theta(g)=\pm g\}\nonumber\\
	=\{X\pm\theta(X)|X\in\mathfrak{g}_{\lambda}\}.
\end{align}
By definition,
\begin{gather}
	m^{+}_{\lambda}=(\mathfrak{g}_{\lambda}\oplus\mathfrak{g}_{-\lambda})\cap\mathfrak{k},\\
	m^{-}_{\lambda}=(\mathfrak{g}_{\lambda}\oplus\mathfrak{g}_{-\lambda})\cap\mathfrak{m},
\end{gather}
and the eigendecomposition yields 
\begin{align}
\mathfrak{g}_{\lambda}\oplus\mathfrak{g}_{-\lambda}=m^{+}_{\lambda}\oplus m^{-}_{\lambda}.
\end{align}
From the construction of $m^{\pm}_{\lambda}$, these spaces have the same dimensions as $\mathfrak{g}_{\lambda}$; that is,
\begin{align}
	\text{dim}(m^{+}_{\lambda})=\text{dim}(m^{-}_{\lambda})=\text{dim}(\mathfrak{g}_{\lambda}).
\end{align}

Summarizing the discussion above, we obtain the following decomposition of the Lie algebra $\mathfrak{g}$:
\begin{align}\label{eq:iwasawadecomp}
	\mathfrak{g}=\mathfrak{h}\oplus Z_{\mathfrak{k}}(\mathfrak{h})\oplus\bigoplus_{\lambda>0}m^{+}_{\lambda}\oplus\bigoplus_{\lambda>0}m^{-}_{\lambda}.
\end{align}
This decomposition Eq.~\eqref{eq:iwasawadecomp} is the Iwasawa decomposition. (In general, the Iwasawa decomposition expresses $\mathfrak{g}$ as the direct sum of the Cartan subalgebra $\mathfrak{h}$, the centralizer $Z_{\mathfrak{k}}(\mathfrak{h})$ in $\mathfrak{k}$ and the remaining subspace. The decomposition Eq.~\eqref{eq:iwasawadecomp} provides a more explicit form of this general decomposition.) Comparing subspaces in Eq.~\eqref{eq:iwasawadecomp} with the definition of $\mathfrak{m}$ and
$\mathfrak{k}$, we obtain
\begin{gather}
	\mathfrak{k}=Z_{\mathfrak{k}}(\mathfrak{h})\oplus\bigoplus_{\lambda>0}m^{+}_{\lambda},\\
	\mathfrak{m}=\mathfrak{h}\oplus\bigoplus_{\lambda>0}m^{-}_{\lambda}.
\end{gather}
Hence their dimensions satisfy
\begin{gather}\label{eq:dimk_z_mplus}
	\dim\mathfrak{k}=\dim Z_{\mathfrak{k}}(\mathfrak{h})+\sum_{\lambda>0}\dim(m^{+}_{\lambda}),\\\label{eq:dimm_h_mminus}
	\dim\mathfrak{m}=\dim\mathfrak{h}+\sum_{\lambda>0}\dim(m^{-}_{\lambda}).
\end{gather}
Using Eq.~\eqref{eq:dimk_z_mplus}-\eqref{eq:dimm_h_mminus} and the relation $\text{dim}(m^{+}_{\lambda})=\text{dim}(m^{-}_{\lambda})$, we obtain
\begin{align}
	\dim\mathfrak{k}+\dim\mathfrak{h}-\dim\mathfrak{m}=\dim Z_{\mathfrak{k}}(\mathfrak{h}).
\end{align}
Since $\dim Z_{\mathfrak{k}}(\mathfrak{h})\ge0$, the inequality Eq.~\eqref{eq:dimmdimhdimk} follows.

\section{\label{appsec:proof_of_int_hamil_general}proof of Theorem.~\ref{thm:int_generalhamil}}
The proof of Theorem.~\ref{thm:int_generalhamil}, which is the Majorana representation (or equivalently, the general fermion-qubit mapping) analogue of Theorem.~\ref{thm:hamalgdim}, is carried out in a similar way in Subsec.~\ref{subsec:proof interact}. 

We note that the Majorana operators satisfy the following anticommutation relations:
\begin{gather}
	\{\gamma_{+,i\sigma},\gamma_{+,j\sigma'}\}=\{\gamma_{-,i\sigma},\gamma_{-,j\sigma'}\}=2\delta_{ij}\delta_{\sigma\sigma'},\\
	\{\gamma_{+,i\sigma},\gamma_{-,j\sigma'}\}=0.
\end{gather}

\begin{lem}\label{lem:ladderparts_majorana}
	For any $0\le i\neq j\le N-1$ and $\sigma\in\{\uparrow,\downarrow\}$, there exists $\alpha\in\{+,-\}$ such that
	\begin{gather}
		\gamma_{+,i\sigma}\gamma_{\alpha,j\sigma},\quad \gamma_{-,i\sigma}\gamma_{-\alpha,j\sigma}\in\mathfrak{g}(H_1,\gamma).
	\end{gather}
\end{lem}
Lemma.~\ref{lem:ladderparts_majorana} corresponds to Lemma.~\ref{lem:ladderparts} when specialized to Jordan-Wigner transformation.

\begin{proof}
	Here we fix the site $i$ and $j$. Since the graph $G$ is connected, there exists a path connecting $i$ and $j$. We denote the path as 
	\begin{align}
		i\to j_1\to j_2\to\cdots\to j_l=j.
	\end{align}
	We prove by induction on $k$ that $\gamma_{+,i\sigma}\gamma_{+,j_k\sigma}$ or $\gamma_{+,i\sigma}\gamma_{-,j_k\sigma}$ belongs to $\mathfrak{g}(H_1,\gamma)$. For the case $k=1$, $(i,j_1)$ is an edge of the graph $G$. Hence $\gamma_{+,i\sigma}\gamma_{-,j_1\sigma}$ is an element of $\mathfrak{g}(H_1,\gamma)$ because the operator in the main text is a generator of the Hamiltonian algebra $\mathfrak{g}(H_1,\gamma)$. Assume, as the induction hypothesis, that the statement holds for $k$, i.e. $\gamma_{+,i\sigma}\gamma_{\alpha_k,j_k\sigma}\in\mathfrak{g}(H_1,\gamma)$ ($\alpha_k=\pm$) belongs to $\mathfrak{g}(H_1,\gamma)$. Since $\gamma_{\alpha_k,j_k\sigma}\gamma_{-\alpha_k,j_{k+1}\sigma}$ belongs to $\mathfrak{g}(H_1,\gamma)$, the commutator of these elements
	\begin{align}
		[\gamma_{+,i\sigma}\gamma_{\alpha_k,j_k\sigma},  \gamma_{\alpha_k,j_k\sigma}\gamma_{-\alpha_k,j_{k+1}\sigma}]\propto \gamma_{+,i\sigma}\gamma_{-\alpha_k,j_{k+1}\sigma}
	\end{align}	
	also belongs to the Hamiltonian algebra $\mathfrak{g}(H_1,\gamma)$. Therefore, for each $k$, there exists $\alpha_k\in\{+,-\}$ such that $\gamma_{+,i\sigma}\gamma_{\alpha_k,j_k\sigma}$ belongs to $\mathfrak{g}(H_1,\gamma)$. In particular, taking $k=l$ and $\alpha=\alpha_l$ completes the proof of Theorem.~\ref{lem:ladderparts_majorana}.

	Using the same path from $i$ to $j$ and following the similar procedure, we conclude that $\gamma_{-,i\sigma}\gamma_{-\alpha,j\sigma}\in\mathfrak{g}(H_1,\gamma)$.
\end{proof}

Hereafter, we define $\alpha(j),\beta(j)\in\{+,-\}$ for each $j$ such that
\begin{gather}
	\gamma_{+,0\uparrow}\gamma_{\alpha(j),j\uparrow},\quad \gamma_{-,0\uparrow}\gamma_{-\alpha(j),j\uparrow}\in\mathfrak{g}(H_1,\gamma),\\
	\gamma_{+,0\downarrow}\gamma_{\beta(j),j\downarrow},\quad \gamma_{-,0\downarrow}\gamma_{-\beta(j),j\downarrow}\in\mathfrak{g}(H_1,\gamma).
\end{gather}

By computing commutators of these operators, it follows that
\begin{gather}
	\gamma_{\alpha(j),j\uparrow}\gamma_{-\alpha(j+1),j+1\uparrow},\quad\gamma_{\beta(j),j\downarrow}\gamma_{-\beta(j+1),j+1\downarrow}\in\mathfrak{g}(H_1,\gamma),
\end{gather}
which corresponds to the later part of Lemma.~\ref{lem:ladderparts}.

we next prove the following lemma which corresponds to Lemma.~\ref{lem:rungparts} in the main text. For convenience, we introduce the notation
\begin{align}
	\tilde{Z}_{i\sigma}:=\gamma_{+,i\sigma}\gamma_{-,i\sigma}.
\end{align}

\begin{lem}\label{lem:rungparts_majorana}
	The following products of Majorana operators belong to the Hamiltonian algebra $\mathfrak{g}(H_1,\gamma)$:
	\begin{gather}\label{eq:zzgg1}
		\tilde{Z}_{0\downarrow}\cdot\tilde{Z}_{i\downarrow}\cdot(\gamma_{\alpha(i),i\uparrow}\gamma_{\alpha(i+1),i\uparrow}),\\
		\tilde{Z}_{0\downarrow}\cdot\tilde{Z}_{i\downarrow}\cdot(\gamma_{-\alpha(i),i\uparrow}\gamma_{-\alpha(i+1),i\uparrow}),\\
		\tilde{Z}_{0\uparrow}\cdot\tilde{Z}_{i\uparrow}\cdot(\gamma_{\beta(i),i\downarrow}\gamma_{\beta(i+1),i\downarrow}),\\\label{eq:zzgg4}
		\tilde{Z}_{0\uparrow}\cdot\tilde{Z}_{i\uparrow}\cdot(\gamma_{-\beta(i),i\downarrow}\gamma_{-\beta(i+1),i\downarrow}).
	\end{gather}
\end{lem}
Lemma.~\ref{lem:rungparts_majorana} is the analogue of Lemma.~\ref{lem:rungparts} when specialized to the Jordan-Wigner transformation.

\begin{proof}
	We define the Majorana-fermion counterparts of the operators $\Sigma_l$ in Eqs.~\eqref{eq:Sigmas1}-\eqref{eq:Sigmas4} of the main text.
	\begin{gather}\label{eq:Sigmas1_majorana}
		\Sigma_{1}(j):=(\gamma_{-,0\uparrow}\gamma_{\alpha(j),j\uparrow})
		(\gamma_{-,0\downarrow}\gamma_{\beta(j),j\downarrow}),\\
		\Sigma_{2}(j):=(\gamma_{-,0\uparrow}\gamma_{\alpha(j),j\uparrow})
		(\gamma_{+,0\downarrow}\gamma_{-\beta(j),j\downarrow}),\\
		\Sigma_{3}(j):=(\gamma_{+,0\uparrow}\gamma_{-\alpha(j),j\uparrow})
		(\gamma_{-,0\downarrow}\gamma_{\beta(j),j\downarrow}),\\\label{eq:Sigmas4_majorana}
		\Sigma_{4}(j):=(\gamma_{+,0\uparrow}\gamma_{-\alpha(j),j\uparrow})
		(\gamma_{+,0\downarrow}\gamma_{-\beta(j),j\downarrow})
		.
	\end{gather}
	To show that these operators lie in $\mathfrak{g}(H_1,\gamma)$, we first take the commutators of $\gamma_{\pm,0\uparrow}\gamma_{\pm\alpha(j),j\uparrow}\text{ and }\tilde{Z}_{0\uparrow}\tilde{Z}_{0\downarrow}$. Then, by taking successive commutators with $\gamma_{\pm,0\downarrow}\gamma_{\pm\beta(j),j\downarrow}$, we obtain that all operators $\Sigma_{1}(j)$-$\Sigma_{4}(j)$ in Eq.~\eqref{eq:Sigmas1_majorana}-\eqref{eq:Sigmas4_majorana} belong to $\mathfrak{g}(H_1,\gamma)$.
	
	To prove that the operators in Eq.~\eqref{eq:zzgg1}-\eqref{eq:zzgg4} lie in $\mathfrak{g}(H_1,\gamma)$, We compute nonvanishing commutators between $\Sigma_{k}(j)$ and $\Sigma_{l}(j+1)$, followed by taking commutators with $\gamma_{\pm,0\uparrow}\gamma_{\pm\alpha(j),j\uparrow}$ or $\gamma_{\pm,0\downarrow}\gamma_{\pm\beta(j),j\downarrow}$. This yields the desired operators.
\end{proof}

To prove Theorem.~\ref{thm:int_generalhamil} in the main text, we follow the same strategy as the argument presented just after the proof of Lemma.~\ref{lem:rungparts}. In the Majorana-fermion representation, the corresponding steps are modified as follows.
We consider the same ladder diagram as in Fig.~\ref{fig:ladderdiagram1}. The modifications are as follows:

\begin{itemize}
	\item \textbf{Step 1.}
	Depending on where the path starts, we choose the initial operator $Q_1$ as follows. If the path starts on the upper lane, we take $Q_1\in\{\gamma_{+\alpha(1),1\uparrow}\gamma_{-\alpha(2),2\uparrow},\gamma_{-\alpha(1),1\uparrow}\gamma_{+\alpha(2),2\uparrow}\}$, and if it starts on the lower lane, we take $Q_1\in\{\gamma_{+\beta(1),1\downarrow}\gamma_{-\beta(2),2\downarrow},\gamma_{-\beta(1),1\downarrow}\gamma_{+\beta(2),2\downarrow}\}$. In each case, we choose whichever of the two operators will anticommute with the operator used at the next step.

	\item \textbf{Step 2.}
	For each $k=1,2,\cdots, N-2$, we construct $Q_{k+1}$ from $Q_k$. The construction depends on how the path traverses the $k$-th junction shown in Fig.~\ref{fig:operation_all}.
	\begin{itemize}
		\item \textbf{Case (a).}
		Among the operators
		\begin{align}
			\gamma_{\alpha(k+1),k+1\uparrow}\gamma_{-\alpha(k+2),k+2\uparrow},\quad\gamma_{-\alpha(k+1),k+1\uparrow}\gamma_{\alpha(k+2),k+2\uparrow},
		\end{align}
		we choose the one that anticommutes with $Q_k$, and multiply $Q_k$ by it (equivalently, we take the commutator with $Q_k$).

		\item \textbf{Case (b).}
		Similarly, among
		\begin{align}
			\gamma_{\beta(k+1),k+1\downarrow}\gamma_{-\beta(k+2),k+2\downarrow},\quad\gamma_{-\beta(k+1),k+1\downarrow}\gamma_{\beta(k+2),k+2\downarrow},
		\end{align}
		we choose the operator that anticommutes with $Q_k$ and multiply $Q_k$ by it.

		\item \textbf{Case (c).}
		We multiply $Q_k$ by $\tilde{Z}_{0\uparrow}\cdot\tilde{Z}_{i\uparrow}\cdot(\gamma_{-\beta(i),i\downarrow}\gamma_{-\beta(i+1),i\downarrow})$ to $Q_k$.
		\item \textbf{Case (d).}
		We multiply $Q_k$ by $\tilde{Z}_{0\downarrow}\cdot\tilde{Z}_{i\downarrow}\cdot(\gamma_{-\alpha(i),i\uparrow}\gamma_{-\alpha(i+1),i\uparrow})$ to $Q_k$.
	\end{itemize}
	The existence of the operators in Cases (a)-(d) is guaranteed by Lemma.~\ref{lem:ladderparts_majorana} and Lemma.~\ref{lem:rungparts_majorana}. We assign to the path the product of Majorana operators $Q_{N-1}$.
\end{itemize}
The above procedure yields $2^{N-1}$ distinct products of Majorana operators. Consequently, the Lie algebra generated by the operators in Eqs.~\eqref{eq:maj_hoppingterm}-\eqref{eq:maj_int} satisfies $\text{dim}(\mathfrak{g}(H_1,\gamma))\ge 2^{N-1}$. This completes the proof of Theorem.~\ref{thm:int_generalhamil}.

\section{\label{appsec:orbitferionqubitmap}The fermion-qubit mapping Eq.~\eqref{eq:orbitalrotation12_1}-~\eqref{eq:orbitalrotation12_3}} 
In this section, we prove the argument that the dimension of the Hamiltonian algebra for the qubit model of the Hamiltonian Eq.~\eqref{eq:model_t_U_decouple} with the fermion-qubit mapping Eq.~\eqref{eq:orbitalrotation12_1}-\eqref{eq:orbitalrotation12_3} grows exponentially. The fermion-qubit mapping Eq.~\eqref{eq:orbitalrotation12_1}-\eqref{eq:orbitalrotation12_3} is represented as the combination of the following spin-orbital transformation,
\begin{gather}\label{eq:sotr_tilde1}
c_{1\sigma}=\frac{1}{\sqrt{2}}(\tilde{c}_{1\sigma}-\tilde{c}_{2\sigma}),\\\label{eq:sotr_tilde2}
c_{2\sigma}=\frac{1}{\sqrt{2}}(\tilde{c}_{1\sigma}+\tilde{c}_{2\sigma}),\\\label{eq:sotr_tilde3}
c_{k\sigma}=\tilde{c}_{k\sigma}\quad(k\ge3)
\end{gather}
and the Jordan-Wigner transformation,
\begin{align}
    \label{eq:jw_fermion_def_dag_tilde}\tilde{c}^{\dagger}_{k\sigma}=Z_1Z_2\cdots Z_{\sigma(k)-1}\left(\frac{X_{\sigma(k)} + iY_{\sigma(k)}}{2}\right),\\\label{eq:jw_fermion_def_tilde}
    \tilde{c}_{k\sigma}=Z_1Z_2\cdots Z_{\sigma(k)-1}\left(\frac{X_{\sigma(k)} - iY_{\sigma(k)}}{2}\right).
\end{align}
Here, we adopt $u(k)=2k-1$ and $d(k)=2k$. By applying the spin-orbital transformation Eq.~\eqref{eq:sotr_tilde1}-~\eqref{eq:sotr_tilde3}, the Hamiltonian Eq.~\eqref{eq:model_t_U_decouple} becomes

\begin{align}
H=\sum_{\sigma}&\left[\frac{-t}{\sqrt{2}}(\tilde{c}^{\dagger}_{1\sigma}\tilde{c}_{3\sigma}+\tilde{c}^{\dagger}_{3\sigma}\tilde{c}_{1\sigma})+\frac{-t}{\sqrt{2}}(\tilde{c}^{\dagger}_{1\sigma}\tilde{c}_{2\sigma}+\tilde{c}^{\dagger}_{2\sigma}\tilde{c}_{1\sigma})\right.\nonumber\\
&\left.-t\sum_{i=3}^{N-1}(\tilde{c}^{\dagger}_{i\sigma}\tilde{c}_{i+1\sigma}+\tilde{c}^{\dagger}_{i\sigma}\tilde{c}_{i+1\sigma})\right]\nonumber\\
+\frac{U}{4}&\left(\tilde{c}^{\dagger}_{1\uparrow}\tilde{c}_{1\uparrow}-\frac12+\left(-\tilde{c}^{\dagger}_{1\uparrow}\tilde{c}_{2\uparrow}-\tilde{c}^{\dagger}_{2\uparrow}\tilde{c}_{1\uparrow}+\tilde{c}^{\dagger}_{2\uparrow}\tilde{c}_{2\uparrow}-\frac12\right)\right)\nonumber\\
&\cdot \left(\tilde{c}^{\dagger}_{1\downarrow}\tilde{c}_{1\downarrow}-\frac12+\left(-\tilde{c}^{\dagger}_{1\downarrow}\tilde{c}_{2\downarrow}-\tilde{c}^{\dagger}_{2\downarrow}\tilde{c}_{1\downarrow}+\tilde{c}^{\dagger}_{2\downarrow}\tilde{c}_{2\downarrow}-\frac12\right)\right)\nonumber\\
=\sum_{\sigma}&\left[\frac{-t}{\sqrt{2}}(\tilde{c}^{\dagger}_{1\sigma}\tilde{c}_{3\sigma}+\tilde{c}^{\dagger}_{3\sigma}\tilde{c}_{1\sigma})+\frac{-t}{\sqrt{2}}(\tilde{c}^{\dagger}_{1\sigma}\tilde{c}_{2\sigma}+\tilde{c}^{\dagger}_{2\sigma}\tilde{c}_{1\sigma})\right.\nonumber\\
&\left.-t\sum_{i=3}^{N-1}(\tilde{c}^{\dagger}_{i\sigma}\tilde{c}_{i+1\sigma}+\tilde{c}^{\dagger}_{i\sigma}\tilde{c}_{i+1\sigma})\right]\nonumber\\
+\frac{U}{4}&\left(\tilde{c}^{\dagger}_{1\uparrow}\tilde{c}_{1\uparrow}-\frac12\right)\left(\tilde{c}^{\dagger}_{1\downarrow}\tilde{c}_{1\downarrow}-\frac12\right)\nonumber\\
&+\cdots,\label{eq:rotatedfermionmodel}
\end{align}

where "$\cdots$" represents the other non-vanishing terms of the products of the fermion operators. Because the hopping term between sites 1 and 3 appears, the graph for the hopping terms of the fermion model written in Eq.~\eqref{eq:rotatedfermionmodel} becomes a connected graph. Hence, the former three terms become the same form as the Hamiltonian $H_1$ in Eq.~\eqref{eq:oneintfermion}. Therefore, the Hamiltonian algebra for the Jordan-Wigner-transformed qubit model of this fermion Hamiltonian has an exponentially growing dimension. 

The remaining terms only create other Pauli strings, which increase the dimension of the Hamiltonian algebra, according to Proposition~\ref{prop:generatedliealg}. Therefore, the dimension of the Hamiltonian algebra for the qubit model obtained by the Jordan-Wigner transformation of the fermion model Eq.~\eqref{eq:rotatedfermionmodel} grows exponentially with the system size.

\bibliography{HamiltonianAlgebraInteractingFermion}

\end{document}